\newif\ifdraft{}
\newif\iffull{}
\newif\ifsubmission{}

\drafttrue{} %%% TURN ON COMMENTS
\fulltrue{}
\iffull{}
\submissionfalse{}
\else
\submissiontrue{}
\fi{}

\documentclass[sigconf,nonacm]{acmart}

\usepackage[utf8]{inputenc}

\usepackage{paralist}
\usepackage[inline,shortlabels]{enumitem}
\usepackage{amsmath,amsfonts,amsthm}
\usepackage{graphicx}
\usepackage{subcaption}
\usepackage{multirow}
\usepackage{multicol}
\usepackage{xcolor}
\usepackage{xspace}
\usepackage{booktabs}
\usepackage{mdframed}
\usepackage{thmtools}
\usepackage{float}
\usepackage{soul}
\usepackage{xcolor}

\usepackage{tikz}
\usepackage{algorithm}
\usepackage[noend]{algpseudocode}

\usetikzlibrary{shapes}
\usetikzlibrary{arrows}
\usetikzlibrary{calc}
\usetikzlibrary{shadows}

\definecolor{myorange}{HTML}{d95319}
\makeatletter
\algnewcommand{\LineComment}[1]{\Statex \hskip\ALG@thistlm {\color{gray}\textrm{// #1}}}
\makeatother
\algnewcommand{\SectionComment}[2]{\Statex {\color{#2}\(\triangleright\) \textrm{#1}}}

\algnewcommand{\InlineComment}[1]{{\hspace{0.5em}\color{gray}\textrm{// #1}}}
\algnewcommand{\InlineCommentText}[1]{{\hspace{0.5em}\color{gray}{// #1}}}
\algnewcommand{\Phase}[1]{\SectionComment{#1}{myorange}}
\MakeRobust{\Call}
\DeclareCaptionFormat{algor}{%
\hrulefill\par\offinterlineskip\vskip1pt%
\textbf{#1#2}#3\offinterlineskip\hrulefill}
\DeclareCaptionStyle{algori}{singlelinecheck=off,format=algor,labelsep=space}
\algnewcommand\algorithmicswitch{\textbf{switch}}
\algnewcommand\algorithmiccase{\textbf{case}}
\algnewcommand\algorithmicassert{\texttt{assert}}
\algnewcommand\Assert[1]{\State \algorithmicassert(#1)}%
\algdef{SE}[SWITCH]{Switch}{EndSwitch}[1]{\algorithmicswitch\ #1\ \algorithmicdo}{\algorithmicend\ \algorithmicswitch}%
\algdef{SE}[CASE]{Case}{EndCase}[1]{\algorithmiccase\ #1}{\algorithmicend\ \algorithmiccase}%
\algtext*{EndSwitch}%
\algtext*{EndCase}%⇧
\algnewcommand\algorithmiccontinue{\textbf{continue}}
\algnewcommand\algorithmicbreak{\textbf{break}}
\algnewcommand\Continue{\algorithmiccontinue}
\algnewcommand\Break{\algorithmicbreak}
\algnewcommand{\IIf}[1]{\State\algorithmicif\ #1\ \algorithmicthen}
\algnewcommand{\EndIIf}{\unskip\ \algorithmicend\ \algorithmicif}
\algblock{As}{EndAs}
\algblock{On}{EndOn}
\algnewcommand\algorithmicas{\textbf{as}}
\algrenewtext{As}[1]{\algorithmicas\ #1}
\algtext*{EndAs}
\algnewcommand\algorithmicon{\textbf{on}}
\algrenewtext{On}[1]{\algorithmicon\ #1}
\algtext*{EndOn}
\iffull{}
\fi{}

\def\protocol{EESMR\xspace}

\def\txpool{\ensuremath{txpool}\xspace}
\def\currentviewvar{\ensuremath{v_{cur}}\xspace}
\def\currentleadervar{\ensuremath{L}\xspace}
\def\currentroundvar{\ensuremath{r_{cur}}\xspace}
\def\Propose{{\sc{Propose}}\xspace}
\def\Blame{{\sc{Blame}}\xspace}

\def\commitUpdateMsg{{\sc{CommitUpdate}}\xspace}
\def\certify{{\sc{Certify}}\xspace}
\def\newViewProposal{{\sc{NewViewProposal}}\xspace}
\def\voteMsg{{\sc{VoteMsg}}\xspace}
\def\Cmds{{\textrm{{\sc{Cmds}}}}\xspace}
\def\blameTimer{\ensuremath{T_{blame}}\xspace}
\def\commitTimer{\ensuremath{T_{commit}}\xspace}

\newcommand{\sizeof}[1]{\ensuremath{\vert #1 \vert}\xspace}
\newcommand{\set}[1]{\ensuremath{\left\{#1\right\}}\xspace}
\newcommand{\p}[1]{\left(#1\right)}

\newcommand{\sign}[1]{\ensuremath{\mathsf{Sign}\p{#1}}}
\newcommand{\vrfy}[1]{\ensuremath{\mathsf{Verify}\p{#1}}}

\def\K{\ensuremath{k}\xspace}
\def\din{\ensuremath{d_{in}}\xspace}
\def\dout{\ensuremath{d_{out}}\xspace}
\def\kcast{\K-cast\xspace}
\def\kcasts{\K-casts\xspace}

\def\h{\ensuremath{\mathcal{H}}\xspace}
\def\V{\ensuremath{\mathcal{N}}\xspace}
\def\E{\ensuremath{\mathcal{E}}\xspace}

\iffull

\else

\fi
\def\parsehgraph{\ensuremath{\h := \p{\V, \E}}\xspace}
\def\parseV{\ensuremath{\V = \set{\node_1,\dots , \node_n}}\xspace}

\def\sendNd{\ensuremath{\mathtt{send}}\xspace}
\def\recvNd{\ensuremath{\mathtt{recv}}\xspace}
\def\delay{\ensuremath{\Delta}\xspace}

\def\node{\ensuremath{p}\xspace}
\newcommand{\nodei}[1]{\ensuremath{\node_{#1}}\xspace}
\def\nodes{\ensuremath{\V}\xspace}
\newcommand{\parseVn}[1]{\ensuremath{\V= \set{\node_1,\dots , \node_{#1}}}\xspace}

\newcommand{\hedge}[2]{\ensuremath{\set{\node_{#1},\left\{#2 \right\}}}\xspace}

\def\block{\ensuremath{B}\xspace}

\newcommand{\blocki}[1]{\ensuremath{\block_{#1}}\xspace}
\def\altblock{\ensuremath{B^{\star}}\xspace}

\newcommand{\altblocki}[1]{\ensuremath{\altblock_{#1}}\xspace}
\def\blkcontent{\ensuremath{b}\xspace}
\def\defaultblkcontent{\ensuremath{\blkcontent_{m}}\xspace}
\def\hashptr{\ensuremath{h}\xspace}
\def\defaulthashptr{\ensuremath{\hashptr_{m-1}}\xspace}

\def\connectivityresult{\ensuremath{\min\limits_{\nodei{i}\in\nodes} \p{ \dout\p{\nodei{i}},  \din\p{\nodei{i}}  }}\xspace}

\def\faulty{faulty\xspace}

\def\opto{\ensuremath{\psi^{\star}}\xspace}
\def\optt{\ensuremath{\psi}\xspace}
\def\optst{\ensuremath{\psi^{Baseline}}\xspace}
\def\opt{\ensuremath{\psi^{\protocol}}\xspace}
\def\nuf{\ensuremath{\nu_f}\xspace}
\def\sigsign{\ensuremath{\sigma_S}\xspace}
\def\sigvrfy{\ensuremath{\sigma_V}\xspace}

\def\macvrfy{\ensuremath{\mu_V}\xspace}
\def\macsign{\ensuremath{\mu_S}\xspace}

\def\blame{\ensuremath{\mathsf{blame}}\xspace}
\newcommand{\msg}[1]{\ensuremath{\left\langle #1 \right\rangle}\xspace}

\def\blameMsg{\ensuremath{\msg{\blame,v}}\xspace}

\def\lockBlock{\ensuremath{B_{lck}}\xspace}
\def\commitBlock{\ensuremath{B_{com}}\xspace}

\renewcommand{\paragraph}[1]{
    \iffull
    \medskip\noindent \textbf{#1.\xspace}
    \else
    \noindent \textbf{#1.\xspace}
    \fi
}
\newcommand{\etal}{\emph{et al.}}

\iffull \fi

\newcommand{\parhead}[1]{\noindent\textbf{#1.}\xspace}
\newcommand{\italhead}[1]{\noindent\textit{#1.}\xspace}

\usepackage[capitalize]{cleveref}

\def\titletext{\protocol{}: Energy Efficient BFT---SMR for the masses}
\title{\titletext{}}

\acmConference[Middleware'23]{Middleware}{December 11-15, 2023}{Bologna, Italy}

\iffull{}
\author{Adithya Bhat}
\email{abhatk@purdue.edu}
\affiliation{%
  \institution{Purdue University}
  \city{West Lafayette}
  \state{IN}
  \country{USA}
}
\author{Akhil Bandarupalli}
\email{abandaru@purdue.edu}
\affiliation{%
  \institution{Purdue University}
  \city{West Lafayette}
  \state{IN}
  \country{USA}
}
\author{Manish Nagaraj}
\email{mnagara@purdue.edu}
\affiliation{%
  \institution{Purdue University}
  \city{West Lafayette}
  \state{IN}
  \country{USA}
}
\author{Saurabh Bagchi}
\email{sbagchi@purdue.edu}
\affiliation{%
  \institution{Purdue University}
  \city{West Lafayette}
  \state{IN}
  \country{USA}
}
\author{Aniket Kate}
\email{aniket@purdue.edu}
\affiliation{%
  \institution{Purdue University}
  \city{West Lafayette}
  \state{IN}
  \country{USA}
}
\author{Michael K. Reiter{*}}\thanks{*The author was also affiliated with Chainlink Labs at the time of this work.}
\email{michael.reiter@duke.edu}
\affiliation{%
  \institution{Duke University}
  \city{Durham}
  \state{NC}
  \country{USA}
}
\fi{}

\begin{document}

\begin{abstract}
Modern Byzantine Fault-Tolerant State Machine Replication (BFT-SMR) solutions focus on reducing communication complexity, improving throughput, or lowering latency.
This work explores the energy efficiency of BFT-SMR protocols.
First, we propose a novel SMR protocol that optimizes for the steady state, i.e., when the leader is correct.
This is done by reducing the number of required signatures per consensus unit and the communication complexity by order of the number of nodes $n$ compared to the state-of-the-art BFT-SMR solutions.
Concretely, we employ the idea that a quorum (collection) of signatures on a proposed value is avoidable during the failure-free runs.
Second, we model and analyze the energy efficiency of protocols and argue why the steady-state needs to be optimized.
Third, we present an application in the cyber-physical system (CPS) setting, where we consider a partially connected system by optionally leveraging wireless multicasts among neighbors.
We analytically determine the parameter ranges for when our proposed protocol offers better energy efficiency than communicating with a baseline protocol utilizing an external trusted node.
We present a hypergraph-based network model and generalize previous fault tolerance results to the model.
Finally, we demonstrate our approach's practicality by analyzing our protocol's energy efficiency through experiments on a CPS test bed.
In particular, we observe as high as $64\%$ energy savings when compared to the state-of-the-art SMR solution for $n=10$ settings using BLE.\@
\end{abstract}

%%% Local Variables:
%%% mode: latex
%%% TeX-master: "main"
%%% End:

\maketitle

% ADD page numbers if not already added
% \pagestyle{plain}

% Introduction Section
\section{Introduction}

Byzantine faults are not new to the computing world.
For more than four decades, the distributed system literature has addressed Byzantine faults using a variety of consensus mechanisms over replicated nodes.
Among those mechanisms, we focus on State Machine Replication (SMR)~\cite{schneiderImplementingFaulttolerantServices1990,castroPracticalByzantineFault2002a, kotlaZyzzyvaSpeculativeByzantine2007a, abrahamSyncHotStuffSimple2020, yinHotStuffBFTConsensus2019, golanguetaSBFTScalableDecentralized2019, bessaniStateMachineReplication2014,shresthaOptimalityOptimisticResponsiveness2020} as it is a generic consensus problem that ensures that the correct nodes agree on the same sequence of values in the presence of Byzantine nodes.
SMR is of general interest in applications such as Multi-party computation~\cite{abrahamBlinderMPCBased2020,bachoNetworkAgnosticSecurityComes2022} which require repeated usages of one-shot consensus problems such as Byzantine Agreement (BA)~\cite{lamportByzantineGeneralsProblem1982} or Reliable Broadcast~\cite{brachaAsynchronousConsensusBroadcast1985}.
An SMR can also implement a distributed ledger (or a blockchain), which has resulted in the tremendous recent interest in these protocols~\cite{abrahamSynchronousByzantineAgreement2019, abrahamSyncHotStuffSimple2020,chanPiLiExtremelySimple2018, castroPracticalByzantineFault2002a, yinHotStuffBFTConsensus2019, golanguetaSBFTScalableDecentralized2019}.

Due to its widespread adoption, energy-efficiency of SMR is an important factor as they are employed in data-centers to ensure system resilience~\cite{singhZenoEventuallyConsistent2009a,ByzantineFailureReal2020}.
They are also considered for the CPS/IoT environments~\cite{novoScalableAccessManagement2019,schillerBlockchainMSP430IEEE2020,profentzasIoTLogBlockRecordingOffline2019,profentzasTinyEVMOffChainSmart2020,liDistributedConsensusAlgorithm2019,bodkheSurveyDecentralizedConsensus2020,yuSpecialIssueResilient2020, zhaoBlockchainEnabledCyberPhysical2021,zhaoBlockchainEnabledIndustrial2019,chooIntroductionSpecialIssue2021} where energy-efficiency can determine how long the system functions.
Even in energy-rich environments such as data-centers, energy costs \emph{can be reduced} by running a more energy-efficient SMR protocol in large scales, and thus this problem is of general interest.

Meanwhile, Cyber-Physical Systems (CPS) provide integrated computational, mechanical, and communication components that interact with each other.
CPS can vary from a system of temperature sensors deployed in power plants~\cite{sheiretovMWMArraySensorsSitu2009}
to a network of soil nutrient monitoring sensors spread over a farm field,
to unmanned aerial vehicles or drones deployed to monitor traffic~\cite{srinivasanAirborneTrafficSurveillance2004}.
In many of these scenarios, the various nodes or devices in such systems need to agree on a common state.
For example, in a military setting, several CPS nodes could be deployed to survey an area; they should maintain a shared state, which they report during their sporadic contacts with a base station.

An adversary may try to intervene in these consensus processes by compromising some nodes: the compromised nodes may communicate incorrect values, possibly in collusion with other compromised nodes, or prevent the correct nodes in the system from performing a consistent action (e.g., Byzantine Generals Problem~\cite{lamportByzantineGeneralsProblem1982}).
While such adversarial threats are obvious in battlefield CPS, they are also becoming pervasive across the entire CPS application spectrum~{\cite{liDistributedConsensusAlgorithm2019,bodkheSurveyDecentralizedConsensus2020,yuSpecialIssueResilient2020, zhaoBlockchainEnabledCyberPhysical2021,zhaoBlockchainEnabledIndustrial2019,chooIntroductionSpecialIssue2021}}.
For example, a recent US Department of Homeland Security (DHS) report on precision agriculture~\cite{mutschlerThreatsPrecisionAgriculture2018} asks to pay particular attention to Byzantine/malicious faults.
The agency observes that an intentional falsification of data can disrupt crop or livestock sectors as the introduction of rogue data into a sensor network can damage crops or herds.

Today, with the focus on Internet-based instantiations, the state-of-the-art SMR solutions try to solely reduce the latency or communication overhead by using, for example, threshold signature schemes~\cite{bonehShortSignaturesRandom2008}.
While theoretically interesting, threshold signatures are computationally intensive and have high energy demands, making SMR unsuitable for CPS.\footnote{Threshold signing operations are also inefficient that many of the protocols, e.g.,~\cite{abrahamSyncHotStuffSimple2020,yinHotStuffBFTConsensus2019,golanguetaSBFTScalableDecentralized2019}, forgo threshold signatures in their implementations, opting instead for multiple traditional digital signatures and resulting linear growth in communication complexity.}
In general, na\"ively employing an SMR protocol that is not optimized for energy efficiency reduces the lifetime of CPS, and we set out to explore energy efficiency for the SMR problem for the CPS setting.

Indeed, the state-of-the-art latency-optimal SMR protocols OptSync~\cite{shresthaOptimalityOptimisticResponsiveness2020} and Sync HotStuff~\cite{abrahamSyncHotStuffSimple2020} are not energy-efficient solutions for SMR due to their high (bit-level) communication complexity and extensive use of expensive cryptographic primitives.
In general, for SMR in CPS, since communication and cryptography are the main sources of energy drain, one might try to use the state-of-the-art SMR protocol, and optimize away the bottleneck.
However, generally optimizing in one direction leads to degradation of the other.
For example, an attempt toward energy efficiency can be to substitute expensive cryptography such as threshold signatures or public-key cryptography with simpler symmetric-key primitives.
Doing so results in increased communication footprint~\cite{aiyerMatrixSignaturesMACs2008,castroPracticalByzantineFault2002a,beerliova-trubiniovaEfficientByzantineAgreement2007} or reduced fault tolerance~\cite{abrahamInformationTheoreticHotStuff2020,castroPracticalByzantineFault2002a,monizIstanbulBFTConsensus2020}, both of which are not ideal.

% Summary of contributions

\paragraph{Contributions}
This work makes the following four contributions towards energy-efficient consensus for CPS:
\begin{enumerate}[(1),wide]
      \item We present \protocol\footnote{{\protocol} stands for {\em Energy Efficient State Machine Replication}. Also pronounced as easy-SMR.}, an energy-efficient, leader-based SMR protocol (\cref{sec:SMR}).
      In the steady state, i.e., when the leader is correct, \protocol{} consumes less energy than the state-of-the-art SMR protocols Sync HotStuff and OptSync.
      The trade-off is that during the view-change phase, i.e., the sub-protocol to change the leader, \protocol{} performs slightly worse than Sync HotStuff~\cite{abrahamSyncHotStuffSimple2020} by adding an extra round.
      Concretely, \protocol{} avoids using computationally expensive certificates and voting in the steady state and pushes those overheads to the view-change phase invoked when a leader stops making progress or if it equivocates.

      \item As every application setting is different, we derive an easy-to-use template for comparing SMR protocols in various CPS settings (\cref{sec:analysis}).
      We model the energy cost of protocols as functions of system parameters such as the number of nodes $n$ and message size $m$.
      Our analysis shows the conditions (number of faults, number of nodes, wireless communication modalities, etc.) where an SMR protocol \textit{A} will be more energy-efficient than a competing protocol \textit{B}.
      We demonstrate the use of this analysis to determine the conditions \protocol{} is desirable compared to a baseline solution where nodes always communicate with a trusted control node.

      \item We observe that in some CPS settings, a node can communicate with its \K neighbors using a single multicast message, instead of \K unicasts with each neighbor.
      We take advantage of such multicasts if they are available, reliable and are more energy-efficient when handling protocol messages than unicasts, to improve the energy efficiency of protocols.
      In order to capture this communication modality, we extend the standard directed graph model of the network using hypergraphs.

      \item Finally, using an experimental CPS test-bed, we analyze the energy efficiency of cryptographic primitives along with parameters as well as the available communication modalities (\cref{sec:results}).
      Moreover, we measure and choose energy-efficient digital signature schemes.
      As \protocol{} and other SMR protocols heavily employ the communication pattern of a single leader signing and the others verifying, we find that the verification-efficient RSA signatures are more energy-efficient than the ECDSA signature scheme.
      We also analyze the energy-vs.-reliability trade-off for Bluetooth Low Energy (BLE) multicasts and demonstrate that \kcasts{} with four nines reliability are more efficient than unicasts.
      For our optimal choice of cryptographic primitives and communication modality, we show that \protocol{} is $2.8$ times more energy-efficient than the state-of-the-art SMR protocol Sync HotStuff~\cite{abrahamSyncHotStuffSimple2020} {in the failure-free runs}, and $2$ times worse than Sync HotStuff during leader changes, demonstrating the positive trade-offs of our approach.
\end{enumerate}

%%% Local Variables:
%%% mode: latex
%%% TeX-master: "main"
%%% End:

%%% Local Variables:
%%% mode: latex
%%% TeX-master: "main"
%%% End:

% Prelims
% Definitions, Notation and Assumptions
\section{Preliminaries}\label{sec-notations}

\parhead{System Model} 
Consider an $n$-node system $\parseV$, with up to $f<n/2$ nodes being Byzantine. 
We assume a static fully-connected point-to-point communication graph. 
The links are bounded synchronous~\cite{wangOptimalGeneralizedByzantine2014,wangReachingTrustedByzantine2014,zimmerlingSynchronousTransmissionsLowPower2020,ferrariEfficientNetworkFlooding2011,sahaEfficientManytoManyData2017,abrahamSynchronousByzantineAgreement2019,abrahamSyncHotStuffSimple2020,abrahamOptimalGoodCaseLatency2022,garayBitcoinBackboneProtocol2015}. 
We assume that $\delay$ is a public upper bound on the message delivery time for any correct sender. 
The adversary can control the delivery of messages as long as the bounded-synchrony assumption is not violated.

\parhead{State Machine Replication --- SMR}
SMR is an abstraction of a state machine that applies input requests from clients to a state and outputs new state. 
An SMR protocol~\cite{schneiderImplementingFaulttolerantServices1990} (defined formally in~\cref{def:SMR def}) ensures all correct nodes implement the SMR abstraction where the nodes reach agreement on an ordering of requests (or transactions) via a linearizable log. 

\begin{definition}[State Machine Replication (SMR)~\cite{abrahamSyncHotStuffSimple2020}]\label{def:SMR def}
    An SMR protocol generates a linearizable log of transactions satisfying the following:
    \begin{asparaenum}[(1)]
        \item \textbf{Safety.} For any position $m$ in the log, if two correct nodes output $\blocki{m}$ and \altblocki{m} respectively, then $\blocki{m} = \altblocki{m}$.
        \item \textbf{Liveness.} Each client request is eventually committed by all the correct nodes.
    \end{asparaenum}
\end{definition}

We make the same assumptions as existing SMR literature~\cite{profentzasIoTLogBlockRecordingOffline2019,profentzasTinyEVMOffChainSmart2020,abrahamSyncHotStuffSimple2020,chanPaLaSimplePartially2018,chanPiLiExtremelySimple2018,shresthaOptimalityOptimisticResponsiveness2020} with respect to the role of clients and the validity of the requests.
Some example assumptions include: all the clients are honest when measuring the throughput, every transaction can be validated for correctness, every protocol message has an identifier for that instance so messages from one instance will not become equivocation\footnotemark{} in another instance, the hash functions are all seeded, PKI is used to setup (possibly threshold) keys before starting the protocol, and the public information is agreed upon by all the nodes as part of the setup before the start of the protocol.

\footnotetext{An equivocation is a Byzantine behavior where a node sends two conflicting messages to two different nodes.}

\parhead{Blocks}
In our instantiation of SMR, we use blocks to denote a unit of the linearizable log.
A block is defined as 
\begin{align*}
    block.contents &= \Cmds, \\
    block.parent &= \text{hash of the parent block}.
\end{align*}
$block.parent$ is the hash of the \emph{parent} block. $block.contents$ is the set of \textit{requests} from clients called \Cmds. 
We let $G$ be the genesis block with height $0$, and recursively define the height of other blocks as the height of the parent plus one.
We refer to the \emph{highest block} as the block with the highest height.

%%% Local Variables:
%%% mode: latex
%%% TeX-master: "main"
%%% End:

\parhead{Clocks}
We assume that the clocks of all the nodes have a constant bounded skew achieved using network time protocol~\cite{sichitiuSimpleAccurateTime2003,sunTinySeRSyncSecureResilient2006,songAttackresilientTimeSynchronization2007} and that all the nodes start within time $\delay$ of each other, which is ensured using a clock synchronization protocol~\cite{abrahamSynchronousByzantineAgreement2019}.
All nodes start at local time $t=0$\footnote{The local time may be off by $\Delta$ across different nodes. This is known as the non-lock step model~\cite{abrahamSyncHotStuffSimple2020} as opposed to the lock-step synchrony model which assumes that all the nodes are in the same round at the same time.}.

\gdef\Din{\ensuremath{D_{in}}\xspace}
\gdef\Dout{\ensuremath{D_{out}}\xspace}

\parhead{Cryptographic Primitives}\label{sec:crypto primitives} 
Generally, there are two classes of authentication primitives: digital signatures which are asymmetric (public-key) primitives, and Message Authentication Codes (MACs) schemes which are symmetric-key primitives. 

A digital signature scheme consists of three algorithms which fail with negligible probability in the security parameter: $\mathsf{KeyGen}$, $\mathsf{Sign}$, and $\mathsf{Verify}$. $\p{sk, pk} \gets \mathsf{KeyGen}(\kappa)$, where $\kappa$ is the security parameter. 
$sk$ is the secret key used to sign the message and $pk$ is the public key used by others to verify that the message was signed by whoever possesses $sk$. 
$\sigma \gets \sign{ sk_i, m }$, where $\sigma$ denotes the signature on $m$ signed by using secret key $sk_i$. 
$\set{0,1} \gets \vrfy{ m,\sigma, pk }$ outputs $1$ if the signature $\sigma$ signed by $sk$ on $m$ is valid, and outputs $0$ otherwise.
We denote $\msg{m}_i$ as $(m, \sigma)$ where $\sigma \gets \sign{sk_i,m}$ and $sk_i$ is the secret key of node $\nodei{i}$.

A Message Authentication Code (MAC) scheme consists of three similar algorithms.
A MAC is generally more energy-efficient than a digital signature, but a digital signature scheme ensures transferable authentication, i.e., a node can convince others that a particular node sent a particular message using a signature on it without possessing $sk$ which is useful for detecting equivocation.  
In a MAC scheme, it is harder to verify equivocation~\cite{aiyerMatrixSignaturesMACs2008}
in the distributed environment even for a moderate number of nodes.

We use both primitives for energy-analysis, and use digital signatures in our protocol.
We assume a public key infrastructure (PKI) to set up digital signatures.
Similar to the related work~\cite{abrahamSynchronousByzantineAgreement2019, abrahamSyncHotStuffSimple2020, shresthaOptimalityOptimisticResponsiveness2020}, threshold signatures can reduce the certificate size to $O(1)$.

%%% Local Variables:
%%% mode: latex
%%% TeX-master: "main"
%%% End:

% Protocol
\section{\protocol{} Protocol}\label{sec:SMR}

% Use the following conventions in the algorithms:
%   := for parsing
%   <- for assignment
%   =  for equality checks

The \protocol{} protocol implements the SMR abstraction (\cref{def:SMR def}).
It builds a chain of blocks linked using hashes, and all the nodes execute the client requests in the blocks in the same sequence.
The clients waits to receive $f+1$ identical acknowledgments with execution results and accepts the results.
We omit the clients and the execution from the rest of the discussion and focus on the nodes that run the SMR.\@

\protocol{} protocol runs in a sequence of \emph{rounds} and \emph{views} with monotonically increasing numbers from $1$.
Each view consists of a dedicated unique leader $L \in \V$ known to all. 
Within a view, the leader proposes blocks in every round.
All nodes maintain pending commands in a local data structure $\txpool$.
The leader proposes blocks using the commands from $\txpool$ and the other nodes on committing a block, remove the commands in the block from the $\txpool$.

\subsection{Data Structures} 

\parhead{Messages and Quorum Certificates} 
Every message $m$ contains the view number in the field $m.view$ and a threshold signature validating it in the field $m.dataSig$ and $m.viewSig$.
The message is customizable into different messages in the protocol by specifying its type in $m.type$ (as specified in \cref{alg:util}).
$n/2+1$ valid threshold signed messages $m$ from the same view and type is combined into a \emph{Quorum Certificate} (using the \Call{QC}{} function in \cref{alg:util}).

\begingroup
\captionsetup[algorithm]{style=algori}
\captionof{algorithm}{Helper functions for \protocol{} (for node $\nodei{i}$).}\label{alg:util}
\small\begin{algorithmic}[1]
\setcounter{ALG@line}{100}
\Function{Msg}{$type$, $data$}
\State{} $m.type \gets type$
\State{} $m.data \gets data$
\State{} $m.view \gets \currentviewvar$
\State{} $m.viewSig \gets \langle m.type, \currentviewvar \rangle_{i}$
\State{} $m.dataSig \gets \langle m.data, \currentviewvar \rangle_{i}$
\State{} \Return{} $m$
\EndFunction{}

\Function{CreateProposal}{$block$, \Cmds}
\State{} $newBlock.parent \gets block$
\State{} $newBlock.contents \gets \Cmds$
\State{} \Return{} $newBlock$
\EndFunction{}

\Function{MatchingMsg}{$m$, $type$, $view$}
\State{} \Return{} $m.type = type \land m.view = view$
\EndFunction{}

\Function{QC}{$V$}
    \State{} $qc.type \gets m.type: m \in V$
    \State{} $qc.view \gets m.view : m \in V$
    \State{} $qc.cert \gets \mathcal{C}\p{m.type, m.view }: \{m.viewSig \mid m \in V\}$
    \State{} \Return{} $qc$
\EndFunction{}

\Function{MatchingQC}{$qc$, $type$, $view$}
\State{} \Return{} $qc.type = type \land qc.view = view$
\EndFunction{}

\Function{LockCompare}{\lockBlock{}, $b$}\label{alg:util:lock update procedure}
\If{$view(\lockBlock) \ne view(b)$ and $b$ extends \lockBlock{}}
\State{} \Return{} $b$
\EndIf{}
\If{$round(\lockBlock) \ne round(b)$ and $b$ extends \lockBlock{}}
\State{} \Return{} $b$
\EndIf{}
\State \Return{} \lockBlock{} \hfill \InlineComment{Retain the old lock}

\EndFunction{}
\end{algorithmic}
\endgroup

%%% Local Variables:
%%% mode: latex
%%% TeX-master: "../main"
%%% End:

\parhead{Book-keeping variables}
We maintain state variables to track the current state of the protocol:
\begin{enumerate*}[(i)]
    \item \currentviewvar{} tracks the current view number, initialized with $1$.
    \item \currentroundvar{} tracks the current round number, initialized with $3$.
    \item \currentleadervar{} tracks the current leader.
    \item \lockBlock{} tracks the currently locked chain, initialized to the genesis block $G$.
    A locked chain is only updated using the \Call{LockCompare}{} function in \cref{alg:util:lock update procedure} and \emph{always} extends the highest committed block.
    A node will never allow any other chain to fork-away from the locked block \lockBlock{} (unless it is safe to do so, described in the view-change) while ensuring that it always extends the highest committed block.
    \item \commitBlock{} tracks the highest committed block, and is initialized to the genesis block $G$.
    \item Timers $\commitTimer(block)$ and $\blameTimer(\currentviewvar)$ are used to commit blocks and send blame messages for a round respectively.
\end{enumerate*}

\subsection{Protocol Overview}
Our protocol is motivated from Sync HotStuff~\cite{abrahamSyncHotStuffSimple2020} and OptSync~\cite{shresthaOptimalityOptimisticResponsiveness2020}.
We present the technical version of our protocol in \cref{alg:protocol}.
A crash-version of our protocol can be obtained by removing the blocks in \cref{alg:protocol:byz1,alg:protocol:byz2} as they are associated with equivocation which is allowed only for Byzantine faults.

The description contains code regions, which are atomically executed.
All blocks such as the conditional \textbf{on} regions are running concurrently.
The same node may act in multiple roles, for instance, a node maybe a leader as well as a node and thus execute both parts of the \textbf{as} regions.

The protocol consists of two sub-protocols: \emph{steady-state protocol} and \emph{view-change protocol}.
The steady-state is a period when the leader is behaving correctly.
The steady-state protocol can be a separate protocol, with safety guarantees but no liveness guarantees.
The view-change is the phase that handles Byzantine behavior of the steady-state leader, by transitioning from view $v$ with the Byzantine leader to view $v+1$.
If the leader of view $v+1$ is Byzantine, then another view-change protocol is initiated.
The view-change protocol needs to ensure both safety and liveness.
Technically, with minor modifications, the view-change protocol can be a stand-alone SMR protocol.
However, it is inefficient, and is only used to bring the system back into a steady-state to enjoy efficiency benefits of the steady-state.

The protocol proceeds in views $1, \ldots$ and rounds $3,\ldots$. 
The first $2$ rounds in the counter are reserved for the view-change protocol, and ensures that the view starts with a safe block.
We assume that the leaders are chosen using the \Call{Leader}{\currentviewvar} function.
This function can be round-robin for simplicity, but for expected constant-latency it is required that the leaders are chosen randomly~\cite{abrahamSynchronousByzantineAgreement2019}.
All messages are validated for correctness, such as formatting, signature checks, etc, and processed when the round is current.
If messages are received from rounds larger than $\currentroundvar$ they are buffered and processed on entering the round to ensure liveness.

\italhead{Note on chain synchronization}
Our description of the protocol accounts for chain-synchronization, i.e., when a node obtains a block and does not know its parent blocks, it will request them from the sender first.
To ease the exposition, we do not include this request-response in the protocol description, but the timers in our description account for these synchronization.
Note that the chain synchronization is not a by-product of our protocol, but is used in protocol and thus does not affect the results.
Since a Byzantine node can trigger chain synchronization once every round, this results in a communication-complexity $O(n^2)$ per round.

\begin{algorithm*}
\caption{\protocol{} Protocol (for node $\nodei{i}$).}\label{alg:protocol}
\begin{multicols}{2}
\small\begin{algorithmic}[1]
\setcounter{ALG@line}{200}
\For{$\currentviewvar \gets 1, 2, 3, \ldots$}
\For{$\currentroundvar \gets 3, 4, 5, \ldots$}
{\Statex{} \hspace{0.6cm} {\color{myorange}\(\triangleright\) \textrm{Steady-State Phase (Safe)}}}
\As{a leader} \hfill \InlineComment{$L \gets$ \Call{Leader}{\currentviewvar}}
\State{} {\color{gray}{//~\Cmds{} is obtained from \txpool{}}}
\State{} $newBlock \gets \Call{CreateProposal}{\lockBlock, \Cmds}$\label{alg:protocol:steady state block stream}
\State{} $propData \gets (newBlock, \currentroundvar)$
\State{} $curProposal \gets$ \Call{Msg}{\Propose{}, $propData$}
\State{} broadcast $curProposal$ once
\EndAs{}
\As{a node} \hfill \InlineCommentText{Also executed by the leader}
\State{} start or reset $\blameTimer(\currentviewvar)$ to time $4\delay$\label{alg:protocol:blame-reset}
\State{} wait for \emph{first} valid $newProposal$ from any node
\LineComment{Vote in the head}
\State{} update $\lockBlock \gets newProposal.data[0]$\label{alg:protocol:line:voting-in-the-head}
\State{} broadcast $newProposal$ once\label{alg:protocol:steady state relay broadcast}
\State{} set $\commitTimer(\lockBlock) \gets 4\delay$\label{alg:protocol:commit-set}
\State{} call \Call{NextRound}{}
\EndAs{}
{\Statex{} \hspace{0.6cm} {\color{myorange}\(\triangleright\) \textrm{View-Change Phase (Safe + Live)}}}

\LineComment{Handle lack of progress}
\On{time out of \blameTimer(\currentviewvar)} \hfill \InlineCommentText{Crash}\label{alg:protocol:blame-np}
\State{} $blameData = \bot$
\State{} $blameMsg \gets$ \Call{Msg}{\Blame, $blameData$}
\State{} broadcast $blameMsg$
\EndOn{}

\LineComment{Handle equivocation}
\On{$pr_1, pr_2 \land pr_1.round = pr_2.round$} \hfill \InlineCommentText{Byz.}\label{alg:protocol:byz1}\label{alg:protocol:blame-eq}
\State{} $blameData \gets \bot$
\State{} $blameMsg \gets$ \Call{Msg}{\Blame, $blameData$}
\State{} broadcast $(blameMsg, pr_1, pr_2)$ once
\EndOn{}

\On{blame $m$ with $m.data = (prop_1,prop_2)$} \hfill \InlineCommentText{Byz.}\label{alg:protocol:byz2}
\State{} cancel all commit timers $\commitTimer(\cdot)$\label{alg:protocol:commit-cancel-eq}
\State{} broadcast $m$ once
\EndOn{}

\LineComment{Change the view}
\On{$f+1$ \Call{MatchingMsg}{$m$, $\blame$, \currentviewvar} as $V$}\label{alg:protocol:quit-view1}
\State{} cancel all commit timers $\commitTimer(\cdot)$
\State{} $blameQC \gets$ \Call{QC}{$V$}
\State{} broadcast $blameQC$ once
\EndOn{}

\On{valid $blameQC$}\label{alg:protocol:quit-view2}
\State{} broadcast $blameQC$ once
\State{} wait $\delay$ \hfill \InlineComment{Ensure all correct nodes quit the view}\label{alg:protocol:quit view wait}
\State{} call \Call{QuitView}{}
\EndOn{}
\EndFor{}
\EndFor{}

\Procedure{QuitView}{}
\State{} $commitData \gets \commitBlock$
\State{} $comReqMsg \gets$ \Call{Msg}{\commitUpdateMsg, $commitData$}
\State{} broadcast $comReqMsg$\label{alg:line:send highest committed block}
\State{} wait $5\delay$ to obtain $commitQC$\label{alg:protocol:new view wait}
\State{} broadcast $commitQC$ and wait $\delay$
\State{} call \Call{NewView}{$commitQC$}

\On{\commitUpdateMsg{} msg $m$ with $b := m.data$}
\If{$b$ does not conflict with \lockBlock}\label{alg:line:vote for highest qc}
\State{} send \Call{Msg}{\certify, $b$} to the sender of $m$
\EndIf{}
\EndOn{}

\On{$f+1$ \Call{MatchingMsg}{$m$, \certify, \currentviewvar} as $V$}
\State{} $commitQC \gets \Call{QC}{V}$
\State{} $commitUpdateMsg \gets $\Call{Msg}{\commitUpdateMsg, $commitQC$}
\EndOn{}

\On{$commitQC'$ for $b'$ from other nodes}
\If{$b'$ does not conflict with \lockBlock and extends $commitQC$}\label{alg:line:update highest QC}
\State{} $commitQC \gets commitQC'$
\EndIf{}
\EndOn{}
\EndProcedure{}

\Procedure{NewView}{$commitQC$}
\State{} $\currentviewvar \gets \currentviewvar + 1$
\State{} $\currentroundvar \gets 1$
\State{} $\currentleadervar \gets$ \Call{Leader}{\currentviewvar}
\As{new leader}
\State{} wait $4\delay$ to hear $commitQC$ from $f+1$ nodes in $status$\label{alg:timer:new view leader wait for QC}
% \State{} update \lockBlock using \Call{LockUpdate}{} function (TODO)
% \State{} $data \gets \{B_{1},\ldots,B_{n/2+1}\}$ where $B_{i}$ are the $n/2+1$ highest locked blocks \lockBlock{}
\State{} $prop \gets$ \Call{Msg}{\newViewProposal, $status$}
\State{} broadcast $prop$ once \hfill \InlineComment{\cref{alg:cond:new proposal first round} calls \Call{NextRound}{}}
\State{} wait for $f+1$ \Call{MatchingMsg}{$m$, \voteMsg, \currentviewvar} with $m.data = H(prop)$ as $V$
\State{} $qc \gets$ \Call{QC}{$V$}
\State{} $propData \gets qc$
\State{} $curProposal \gets$ \Call{Msg}{\Propose, $propData$}
\State{} broadcast $curProposal$ once
\EndAs{}
\As{a node}
\State{} send $commitQC$ to \currentleadervar
\State{} set $\blameTimer(\currentviewvar)$ to $8\delay$\label{alg:protocol:view-change:short-circuit}
\EndAs{}

\On{valid \newViewProposal{} msg $\land \currentroundvar = 1$}\label{alg:cond:new proposal first round}
\State{} broadcast $prop$ once
\If{$prop.data$ extends highest block in $prop.data.status$}
\State{} $data \gets H(prop)$
\State{} $voteMsg \gets$ \Call{Msg}{\voteMsg, $data$}
\State{} broadcast $voteMsg$
\State{} $\blameTimer(\currentviewvar) \gets 6\delay$\label{alg:protocol:view change round 2 blame}
\State{} call \Call{NextRound}{}
\EndIf{}
\EndOn{}

\On{valid \Propose{} msg $\land \currentroundvar = 2\land$ valid QC in $msg.data$}
\State{} broadcast $msg$
\State{} call \Call{NextRound}{} \hfill \InlineComment{Go to steady-state}
\EndOn{}
\EndProcedure{}

\Phase{Commit Rule}\hfill \InlineCommentText{Anytime}
\On{time out of $\commitTimer(block)$}\label{alg:protocol:commit-rule}
\State{} update \commitBlock{} to $block$
\State{} commit block and its parents
% \State{} update \commitBlock{}
\EndOn{}
\end{algorithmic}
\end{multicols}
\end{algorithm*}

%%% Local Variables:
%%% mode: latex
%%% TeX-master: "../main"
%%% End:

\subsection{\protocol{} --- Steady State}\label{sec:steady state}

In the steady-state, the leader proposes blocks for every round.
These blocks are broadcast to all the nodes (\cref{alg:protocol:steady state block stream}).
The leader continuously streams proposals allowing the protocol to use all available bandwidth.
The other nodes receive the proposal before the timer (\cref{alg:protocol:blame-np}) runs out.
On receiving a new block, all nodes also update their locked block \lockBlock{}.
All nodes broadcast this proposal to all the other nodes (\cref{alg:protocol:steady state relay broadcast}).

\italhead{Committing}
For a block $B$ to be committed~\cref{alg:protocol:commit-rule,alg:protocol:commit-set}, the node waits for $4\delay$ time using $\commitTimer(B)$ after broadcasting $B$ to all the nodes.
During this wait, the node ensures that it does not hear any equivocating blocks for that view and round.
When the $\commitTimer(B)$ expires, the node also tracks the highest committed block in the variable \commitBlock{}.

The $4\delay$ wait ensures that if a node commits a block at time $t$, then the nodes must have forwarded it to all the nodes at time $t-4\delay$ which was received by all the correct nodes by time $t-3\delay$.
After chain-synchronization, all nodes would have this block by time $t-\delay$.
Since, no equivocation was heard by time $t$, all the correct nodes will \emph{lock} on this block, i.e., the \lockBlock{} of all correct nodes will always extend $B$.
The view-change ensures that the first block in the new view always extends the highest committed block of all correct nodes, thus ensuring safety for $B$ even if a view-change occurs.
We discuss the details of this in the description of the view-change protocol.

\italhead{Voting in the head}
In the steady-state, by waiting for $4\delay$ time before committing a block $B$ and not hearing any equivocating blocks, a node has implicitly collected votes for $B$ from all the correct nodes.
We describe this as \emph{voting in the head}, as opposed to explicit quorum certificate construction.
These implicit votes are made explicit during a view-change.
On a high-level, in the view-change protocol, the correct nodes send their vote for blocks that extend their locked block \lockBlock{}.
Due to hash-chaining of blocks, a vote for any child block is also a vote on all its ancestors, and thus all the committed blocks obtain explicit votes during the view-change.

Firstly, this results in $O(1)$ signing operations per node for every committed block in the steady-state, which contributes to energy-savings.
Secondly, since this is a re-broadcast of the leader's proposal and not broadcast of explicit vote messages, in partially connected networks, it results in a flooding of only one message reducing the communication complexity of our protocol.

\subsection{\protocol{} --- View Change}\label{sec:prot view change}

In \protocol, a view change can occur if the leader of a view
\begin{enumerate*}[(i)]
    \item allows $\blameTimer$ to time out for at least $1$ correct node,
    \item equivocates in a round with two blocks at the same height, or
    \item does not correctly extend the highest committed block during the view change (rounds $1$ and $2$)
\end{enumerate*}.

If a correct node $\nodei{i}$ does not receive the next proposal on time, it sends a blame message (\cref{alg:protocol:blame-np}).
If no correct node hears a block then there will be $n-f>f$ blames, and a view change will occur.
But, if a correct node $\nodei{i}$ blames, it may not always result in a view-change, as less than $f+1$ correct nodes may send this, but it is still safe.
If other correct nodes hear a proposal for that round, then they will forward it to node \nodei{i} within \delay{} and this results in a reset of the blame timer (\cref{alg:protocol:blame-reset}).

A leader $L$ can equivocate by sending \block{} and $\block'$ for any round (not just the latest round).
From bounded-synchrony, we can ensure all correct nodes will hear it in time \delay{}.
On observing two equivocating (conflicting) proposals (\cref{alg:protocol:blame-eq}), all correct nodes $n-f>f$ send blame messages and also cancel their commit-timers (\cref{alg:protocol:commit-cancel-eq}) to preserve safety.

\italhead{Quitting the View}
On receiving $f+1$ blame messages for \currentviewvar{} (\cref{alg:protocol:quit-view1}), the nodes build a certificate and broadcast it, signaling other correct nodes to quit the view \currentviewvar{} (\cref{alg:protocol:quit-view2}). 
After sending the blame certificate, correct nodes wait for $\delay$ to ensure that all correct nodes quit the old view $v$.

After quitting the view \currentviewvar, the nodes first obtain certificates for their highest committed block \commitBlock{}.
They do this by broadcasting their \commitBlock{}.
All correct nodes vote for other's \commitBlock{} if it does not conflict with the node's local \commitBlock{} and \lockBlock{}.
As we show in the proofs, this will always be true for all correct nodes.

The nodes then broadcast their highest certificate to all the nodes.
The other node's update their local highest certificate if the newly received certificates are longer and does not conflict with their local \lockBlock{}.
This ensures that no matter whose certificate the (potentially Byzantine) leader picks along with $f$ other Byzantine server's certificates, it must extend the highest committed block of all the correct nodes.
The final highest certificate is input to the next function \Call{NewView}{}.

In this step, the nodes wait $5\delay$ time before starting the next view to ensure that all correct nodes have sufficient time to collect a certificate for their \commitBlock{}.
Say, the first correct node quits the view \currentviewvar{} at time $t$.
By time $t+\delay$, all correct nodes will quit the view \currentviewvar{}.
All correct nodes will have broadcast their \commitBlock{} by this time, which will reach all correct nodes by time $t+2\delay$, which will have been chain synchronized by time $t+4\delay$.
By time $t+5\delay$, all correct nodes will have collected $f+1$ votes for their \commitBlock{}.

\italhead{Starting the New View}
In this function, all correct nodes send their certificates for \commitBlock{} to the new leader.
The new leader collects certificates from all the nodes and proposes a block containing $f+1$ such certificates.\footnote{This step can be optimized to reduce the size, if there are $f+1$ certificates for the same block.}
The nodes vote for this proposal if it extends the highest block among all the certificates.

A correct leader can always build a valid proposal.
Say, the leader is the first node to enter the new view at time $t$.
By time $t+\delay$, all correct nodes will have entered the new view and sent their certificates.
These reach the leader by time $t+2\delay$, and accounting for chain synchronization, the leader can construct a proposal by time $t+4\delay$.
Then it will propose a \newViewProposal{} message for the round $1$ in the new view.
It then builds a quorum certificate and proposes the round $2$ message, after which it transitions to the steady state.

While the new leader is collecting the highest certificates from all the other correct nodes, we need to ensure that the other correct nodes wait for sufficient time so that they do not blame an correct leader.
Otherwise, it can result in liveness problems, as the system changes views all the time.
Let the earliest node enter the new view $v+1$ at time $t$.
The latest correct node (possibly the leader) will enter the view $v+1$ by time $t+\delay$ (\cref{alg:protocol:quit view wait,alg:protocol:new view wait}).
The leader will receive the latest locked blocks from all the correct nodes by time $T+5\delay$ and it will propose a block at this time which will reach all the correct nodes by time $T+6\delay$.
Adding $2\delay$ time for chain synchronization, if the leader is correct, it suffices to wait for $8\delay$ time in \cref{alg:protocol:view-change:short-circuit}.
By time $T+9\delay$, the leader can collect votes from all correct nodes and propose.
This proposal will reach all the correct nodes by time $T+10\delay$ and time $T+12\delay$ after synchronization.
However, some correct nodes may have started their timer for the proposal in round $2$ of the new view at time $T+6\delay$.
Therefore, for the second round, a timer of $6\delay$ will suffice in \cref{alg:protocol:view change round 2 blame} to prevent blaming a correct leader.

If the new leader is Byzantine, as long as it extends one correct node's commit certificate, safety is guaranteed.
The intuition is that this certificate is at least the highest committed block from the previous view(s).
If the new leader does not make progress in round $1$ or fails to produce a valid quorum certificate in round $2$, then the timers are triggered, and another view-change protocol begins.
This preserves liveness in the protocol.

The drawback of the savings during the steady-state is that during the view-change, all the nodes need to generate certificates for their committed blocks.
In this step, the nodes convert their votes in the head to explicit votes.

\subsection{Discussion}
In this section, we present optimizations, analysis, and general discussions about protocol and model.

\parhead{Equivocation scenario speedups}
We can speed up the view change in the event of an equivocation by skipping the wait to construct a quorum certificate. 
This follows from the fact that all correct nodes will vote for the equivocation as the signatures must match for all the correct nodes. 
We can thus save communication in this phase.

\parhead{Signature verification optimization}
Cryptographic operations are typically expensive. 
Observe that the protocol in the blame phase can work without all nodes verifying the equivocation from the leader. 
For systems with more than $2f$ nodes, say $n = 2f+c$, it is sufficient for $2f+1$ nodes to verify the signature and broadcast a blame on efficient media such as TLS if the leader equivocates. 
If a node receives $f+1$ such messages, then the node knows that at least one node has detected the equivocation by the leader and can send a \blameMsg{} message.

\parhead{Batching optimization}
Batching techniques often aid in amortizing the cost of consensus. In \protocol{} protocol, the nodes can optimistically pre-commit to the block received from the neighbors without signature checks. 
After $c$ such rounds, the nodes can initiate a checkpoint protocol where a full SMR protocol is initiated. 
If the leader was correct, this saves a significant amount of energy for the system. 
If the leader was faulty for all $c$ rounds, the nodes need to recover blocks for those rounds. 
The worst-case scenario is the same as the standard \protocol{} protocol. 
However, we have a significant energy improvement in the best-case.

\parhead{Add commands in rounds $1$ and $2$}
For ease of exposition, client commands \Cmds{} are not included in rounds $1$ and $2$ of the view change.
This can be added to improve throughput slightly.

\parhead{On optimistic responsiveness} For applications such as CPS, we assume that the real network speed $\delta$ satisfies $\delta\approx \delay$ due to the scheduling of messages by the nodes.
However, if this is not the case (i.e., $\delta\ll \Delta$), then Sync HotStuff~\cite{abrahamSyncHotStuffSimple2020}, OptSync\cite{shresthaOptimalityOptimisticResponsiveness2020}, and rotating BFT SMR~\cite{abrahamOptimalGoodCaseLatency2022} support \emph{optimistic responsiveness}, a mode where the commit latency is $2\delta$ (using certificates for every block). 
In such settings, our \protocol{} presents a trade-off: by providing a block period of $0$ and being throughput-centric (while being energy-efficient) with Sync HotStuff or other trusted control node based mechanisms who are latency-centric.

\paragraph{Extensions to BA and BB} It is tempting to use the technique of lack of equivocation within $4\delay$ to implement a Byzantine Broadcast protocol~\cite{peaseReachingAgreementPresence1980}. 
This is non-trivial since a Byzantine leader can prevent termination in a run by sending equivocations so that only some correct nodes terminate. 
The correct nodes cannot terminate even when the leader is correct, as this run is indistinguishable from the previous run. 
Using termination certificates, we can extend this work to implement BB and BA similar to Abraham et al.~\cite{abrahamSynchronousByzantineAgreement2019}. 
The benefits of such an approach in this case is limited to the reduction of usage of certificates in the first iteration only.

\parhead{Taking advantage of control planes}
For instance, CPS nodes may be managed by online trusted control nodes and
if such a trusted node is available, we can use it during the view-change to simplify and improve energy-efficiency.
Most SMR protocols can drop liveness and only preserve safety, i.e., the nodes will ensure that if two correct nodes output blocks $B$ and $B'$ at height $h$, then $B=B'$. By forgoing liveness, i.e., if two correct nodes output two blocks $B$ and $B'$ at height $h$ then $B=B'$; this does not state that all correct nodes \emph{will} output $B$ at height $h$, when the leader is correct, we can guarantee safety. When the leader equivocates or does not make progress, the nodes can simply stop and wait for the control nodes to change the view. This can improve energy-efficiency.

\parhead{Security Analysis}
We present the detailed security analysis in 
\iffull{}\cref{sec:security}\fi{}
\ifsubmission{}\cite{fullversion}\fi{}.

\parhead{Note on optimizations and security}
It is easy to show that these optimizations do not violate the safety and liveness of the system. 
We can easily extend the proofs provided in 
\iffull{}\cref{sec:security}\fi 
\ifsubmission{}\cite{fullversion}\xspace{}\fi{} to show that the protocols with the optimizations are still
secure.

%%% Local Variables:
%%% mode: latex
%%% TeX-master: "main"
%%% End:

% Analysis

\section{Energy Analysis}\label{sec:analysis}

In this section, we formalize the energy-measurement model, argue for the need to optimize the common-case and push complexities to the worst-case.
We also develop a framework to compare protocols that allows us to make better decisions.

\parhead{The need for best-case optimality}
Let $\vec{X}$ define a vector, which consists of system parameters such as $\vec{X} = (n, f, m, S$, $R, \sigma_s, \sigma_v)^T$, where $n$ is the total number of nodes tolerating $f$ faults; $m$ is the maximum supported payload size, e.g., the size of \Cmds{}; $S$ and $R$ are the costs to send and receive per byte; and $\sigma_s$ and $\sigma_v$ are the costs to sign and verify digital signatures.

Let $\psi$ denote the energy cost function of a leader based protocol per unit of consensus (typically a block of client requests).
The function takes a column vector $\vec{X}$ of energy-costs of primitives and computes the energy cost of the protocol.
An example is: $\psi(\vec{X}) = c_1m + c_2n + c_3mn + c_4mnS + c_5mnR + c_6\sigma_s + c_7n\sigma_v$ for some constants $c_1$, $c_2$, $c_3$, $c_4$, $c_5$, $c_6$ and $c_7$.
This function ignores the common state machine replication costs such as verifying the semantic validity of client requests, checking if the client has correctly generated the requests, execution of a request, etc.
When comparing two different protocols, we drop $\vec{X}$, and just use $\psi$.
In later section, we use such functions to model protocols and perform energy analysis.

We denote a particular protocol by placing its name in the superscript (e.g., $\optst$), the best-case energy-cost (without faults) of the protocol using subscript $B$ such as $\optt_B$, the worst-case (with faulty leader and faults)\footnote{By worst case, we ignore a denial-of-service adversary that spams invalid messages. Such adversaries are out-of-scope of this work and addressing them is an orthogonal problem~\cite{abrahamSynchronousByzantineAgreement2019,abrahamSyncHotStuffSimple2020,abrahamOptimalGoodCaseLatency2022,chaiByzantineFaultTolerant2014}. For example Elmamy et al.~\cite{elmamySurveyUsageBlockchain2020} classify DoS and jamming as a network-layer concern and not application layer (where we design protocols).} energy cost
and the view change energy costs (costs to change the leader) using the subscripts $W$ and $V$ respectively such as $\optt_W$ and $\optt_V$.
We assume $\optt_V = \optt_W - \optt_B$\footnote{Note that we do not assume that $\optt_V > 0$. There exist protocols where the protocol aborts early in case of Byzantine failures such as Tendermint~\cite{buchmanLatestGossipBFT2019}, which results in $\optt_V < 0$.}\footnotemark{}.

\footnotetext{What about cases that are neither best or worst-case? Examples of such cases are nodes that send incorrect messages, spam messages, etc. We ignore such cases as they are an orthogonal problem common to all SMR protocols and can be resolved by heuristics such as reputation or banning mechanisms for such peers.}

We wish to design a protocol $\optt$ that is more energy efficient than another protocol $\opto$.
Trivially, if $\optt_B \le \opto_B$ and $\optt_W \le \opto_W$, then we have already achieved our goal.
Therefore, let us consider the case where $\optt \le \opto$ is not true always.
Let $N$ be the total number of blocks agreed upon in which $W$ worst-case events occur which result in $V = W$ view changes.
Then using $\p{N-V}\optt_B + V\p{\optt_W} \le \p{N-V}\opto_B + V\p{\opto_W}$, we get
  $\nuf = \tfrac{V}{N} \le \frac{\opto_B-\optt_B}{\optt_V-\opto_V}$.
Let $\nuf =\tfrac{V}{N}$ be the ratio of view changes $V$ to total number of
blocks $N$.
The fraction of \emph{good} (best-case) runs are $1-\nuf$
with $0 \le \nuf \le 1$.

\parhead{(Un)Favorable conditions}
There are two regions of solutions for the above inequality when $\nuf\ge 0$:
\begin{inparaenum}[(i)]
	\item $\optt_B > \opto_B$ (worse best-case) with $\optt_V < \opto_V$ (better
	worst-case), which we term as the \textit{worst-case optimal} solution set,
	and
	\item $\optt_B < \opto_B$ (better best-case) with $\optt_V > \opto_V$ (worse
	worst-case), which we term as the \textit{best-case optimal} solution set.
\end{inparaenum}

This means that we want to either make the protocol worst-case efficient (region i) or best-case efficient (region ii) for $\optt$ to be more energy efficient than $\opto$.
In bounded-synchrony, since the number of faults are bounded by $f$, $V$ is also bounded by $f$.
Therefore, for our setting, asymptotically we need a best-case optimal $\optt$, i.e.,\ the \textit{best-case optimal} solution set.
We also want $\optt_V-\opto_V$ to be as close as possible.
In other words, if $\optt_V$ is much better than $\opto_V$, then the number of best-case events $N-V$ we need to run the protocol will be large, in which case a worst-case optimal solution set may be better.
\[ N \ge V\p{\dfrac{\optt_V-\opto_V}{\opto_B-\optt_B}} \]

For non-synchronous systems, observe that the adversary can choose which phase the protocol spends the most time on.
As a result, an adversary with the goal of thwarting energy-efficiency can always do so by ensuring that the protocol is spending most of its time in the non-optimal phases.
Hence, we restrict ourselves to bounded-synchronous systems in this work and leave the optimization problem for non-synchronous systems an interesting open problem for the future.

For any network setting, Dolev and Strong~\cite{dolevAuthenticatedAlgorithmsByzantine1983} showed in \cref{thm:DS lower bound} that in the worst-case it is impossible to avoid the $f+1$ round lower bound.
For our synchronous setting, the network is always reliable within the $\delay$ parameter.
In a leader based protocol we will have at most $f$ Byzantine leaders.
Thus, intuitively, protocols must be designed to be best-case optimal.

\begin{theorem}[Lower bound on BA~\cite{dolevAuthenticatedAlgorithmsByzantine1983}]\label{thm:DS lower bound}
	Byzantine agreement using authentication can be achieved for $n$ processors
    with at most $f$ faults within $f+1$ phases, assuming $n > f+1$.
\end{theorem}

This bound reiterates that we need to make protocols energy efficient in the
best-case scenario with lower view-change costs to maximize $f_e$, compared to
the baseline justifying the fit of \protocol{} when compared to other protocols.

In practice, protocols may exist where a view-change may exceed the energy available (say the energy available in a fully-charged battery).
In such conditions, the worst-case optimal solution might be more practical, or energy-efficiency may not possible for the given $f$ faults and the given energy-budget.

\parhead{Adversary for Energy-Efficiency}
Let $f_s$ be Byzantine nodes that attack safety.
These may send incorrect messages, but send messages when they must, i.e.,\ they do not attack liveness.
Let $f_c$ be Byzantine nodes that attack liveness (they do not send messages, but if they do, they will send the correct messages).
For safety, we have $n > f_s$, and for liveness we have a necessary condition $n>2f_c$.

We explore a similar relationship for energy faults $f_e$, i.e.,\ faults that wish to drain the energy of correct nodes.
If we want a protocol to be energy fault-tolerant, then the number of faults $f_e \gets f$ should be such that $f_e$ of the worst-cases for $\optt$ should still be better than $f_e$ of the worst-cases for $\opto$.
Solutions for $f$ that satisfy $\optt(f,\cdot) \le \opto(f,\cdot)$ dictate the energy bound.
This is true for example, when $\optt$ and $\opto$ are monotonic in $f$.
% i.e.,\ $\pd{\optt}{f} \le \pd{\opto}{f}$.
For our setting, let $f_e$ Byzantine nodes cause the worst case scenario to occur $f_e$ times after which finally a best-case scenario occurs.
Then using $f_e\cdot\opt_W +\opt_B \le \optst$, we obtain
\begin{align}
f_e &\le \dfrac{\optst - \opt_B}{\p{\opt_B + \opt_V}}
\tag{EB}\label{eq:energy bound}
\end{align}

In many applications, despite the main protocol tolerating up to $50\%$ faults, the actual protocol might tolerate $f<f_{\max}$ due to connectivity bottlenecks or other reasons.
Obtaining $f_e$ in relation to another protocol helps to choose the more energy-efficient protocol.

Our analysis allows administrators and protocol stake-holders and deployers to model protocols and use the application details such as expected number of faults, communication and computation primitives used to make energy-aware protocol choices.

%%% Local Variables:
%%% mode: latex
%%% TeX-master: "main"
%%% End:

% Experiments
%

\section{CPS and Performance Analysis}\label{sec:results}

In this section, we consider the Cyber-Physical System (CPS) context to evaluate the energy efficiency of our protocol.
Our choice of CPS is motivated by two different observations. a) The energy-starved nature of conventional CPS nodes makes it essential to ensure that the energy overhead of a BFT protocol is minimal on the system. b) Conventional machines run a heavy operating system with millions of CPU instructions. Calculating the energy consumed by the protocol from the machine's overall energy consumption requires us to separate the operating system's inherent energy consumption, which is prone to a lot of noise. This noise masks the energy consumed by the BFT protocol. Hence, we consider low-powered CPS nodes, which run a lightweight OS with much lower energy consumption. Consequently, the noise contributed by the OS to the energy readings of the device is also lower, allowing us to measure the energy of the protocol accurately.

\subsection{Protocol Energy Comparison}

\begin{figure}[!t]
    \centering
    \includegraphics[width=0.90\linewidth]{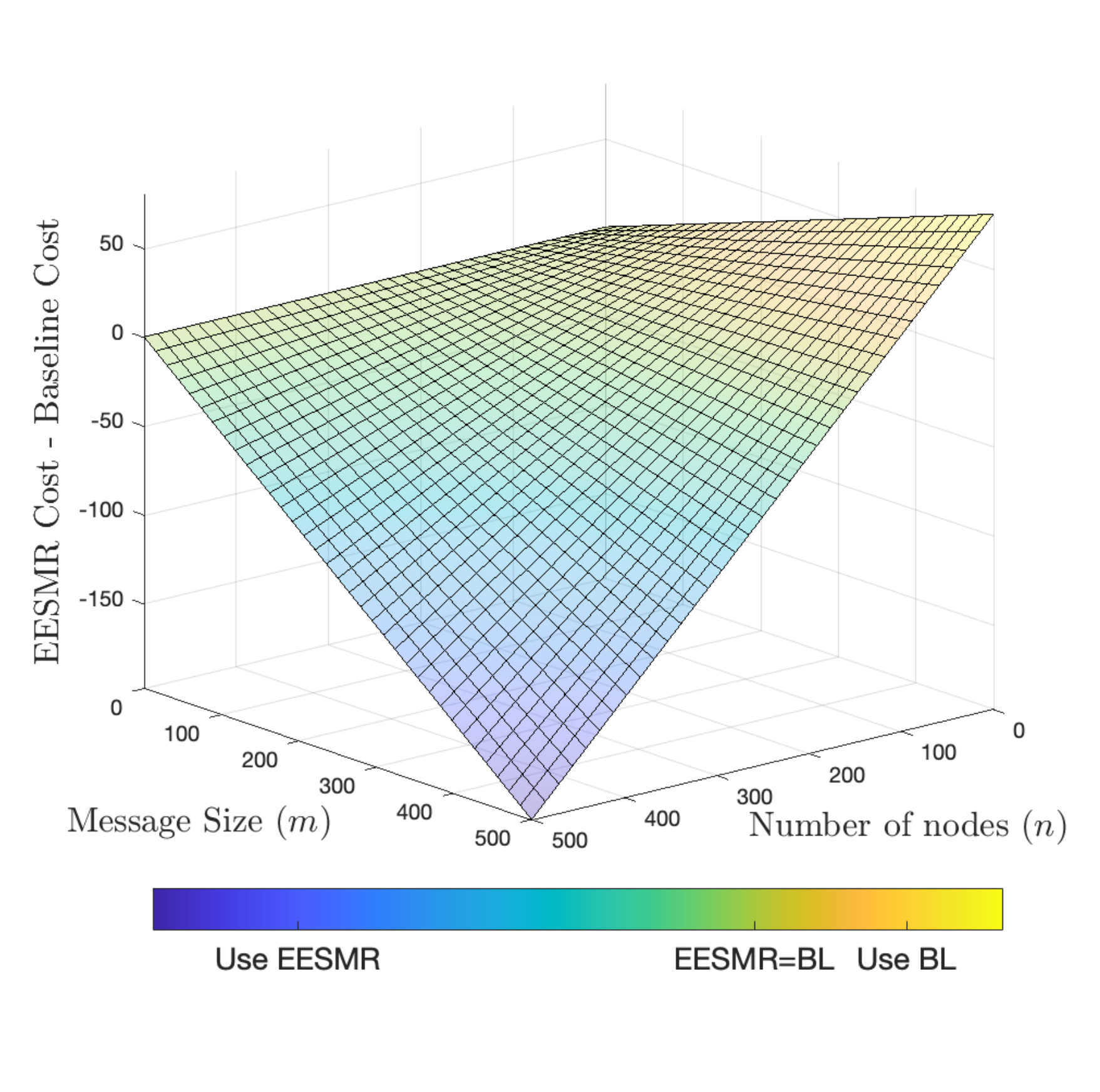}
    \vspace{-7mm}
    \caption{\ifsubmission\small\fi{\textbf{An example feasible region plot} for different message sizes $(m)$ and number of nodes $(n)$. This plot is generated for RSA-1024, and the nodes communicate with each other using Wi-Fi and the external trusted node uses 4G. \textbf{BL} is Baseline.} The z-axis is the difference between \protocol{} energy cost and the baseline. When negative,	\protocol{} is more energy-efficient.}\label{fig:e2c feasible region}
	\ifsubmission\Description{}\fi
\end{figure}

In the previous sections, we designed our \protocol{} to be more energy-efficient than the state-of-the-art protocol Sync HotStuff.
However, every application is unique and may contain several communication media along with trusted control nodes.
\protocol{} may not be the most optimal solution in all of these scenarios.
In this part, we present analysis techniques to compare energy-efficiency of protocols.

\paragraph{Comparison with trusted-baseline}
We compare our work with a trusted-baseline protocol using theoretical tools.
In this baseline protocol, we assume the existence of a trusted node.
Trivially, if the trusted node uses the same communication network as the other nodes, then the trusted-baseline is always energy-efficient.
We consider the case where the trusted node is on a different (more expensive) communication medium than the communication medium among nodes (e.g., satellites to communicate with the trusted node and ground-stations for local communication among other nodes).
An example of such scenario is communication between CPS/IoT devices and the control servers.

The baseline protocol assumes that all the CPS nodes are directly connected to the trusted node using the expensive medium and not use the links between the CPS nodes.
The baseline still assumes that $f$ out of $n$ CPS nodes are Byzantine.

We analytically count operations and build equations for \optst, $\opt_B$ and $\opt_V$ in MATLAB software.
We use the energy measurement values observed by related works and from our experiments in \cref{sec:exp:primitives} in our analysis.
We present an example feasible region in \cref{fig:e2c feasible region} which shows when \protocol{} is favorable over a trusted-baseline protocol (BL) assuming an external trusted party over a slightly more expensive medium 4G, while the CPS nodes use Wi-Fi to run \protocol{}.
Such modeling is useful for example, to choose the more efficient protocol for a given set of parameters.
This analysis can be extended to consider our crash-fault tolerant version, or our steady-state can be combined with the trusted-baseline only for the view-change to improve the energy-efficiency.

\subsection{CPS System Model}
\parhead{Multicasts}
We take advantage of the multicasts available in the CPS setting by modeling the network as hypergraphs (\iffull{}\cref{sec: fault tolerance result}\else{}\cite{fullversion}\fi{}).
The wireless multicast for instance guarantees that all correct nodes hear a consistent value in a time-bounded fashion (see note on jamming).
We assume that all nodes have access to a multicast channel with at least $k$ nodes such that the system remains $f$-connected, i.e., every node can connect with $f+1$ nodes.

\parhead{Notations}
We refer to these multicast channels as \kcasts as we can multicast once to reach $k$ nodes.
These are realized using Ethernet multicasts or wireless multicasts in the CPS system.
Our protocols are updated accordingly, to use these \kcasts instead of broadcast, and an appropriate $\delay$ parameter is used to ensure that all correct nodes receive a message sent by any other correct node.
We denote by $\Din$ and $\Dout$ the number of incoming and outgoing \kcast links.

\italhead{Note on Bounded-Synchrony for CPS} The employed equivocation detection is only applicable to bounded-synchronous (CPS) networks, which
is viable in closed controlled environments~\cite{wangOptimalGeneralizedByzantine2014, wangReachingTrustedByzantine2014, zimmerlingSynchronousTransmissionsLowPower2020, ferrariEfficientNetworkFlooding2011, sahaEfficientManytoManyData2017, jiangHybridLowPowerWideArea2021, chatterjeeContextAwareCollaborativeIntelligence2021} such as remote, industry and military applications.
Notice that for asynchronous and partially synchronous network models, due to the unknown delays, an adversary can always deplete any finite energy resources of the CPS devices, making any meaningful energy analysis challenging.
The energy-efficiency analysis of such protocols during periods of synchrony is of independent interest and presents an interesting future work.

\italhead{Note on Jamming} We assume that the hyper-edges are synchronous~\cite{jiangHybridLowPowerWideArea2021,chatterjeeContextAwareCollaborativeIntelligence2021,deifAntColonyOptimization2017,wanPSFQReliableTransport2002,wagenknechtSNOMCOverlayMulticast2012,wangOptimalGeneralizedByzantine2014,wangReachingTrustedByzantine2014,zimmerlingSynchronousTransmissionsLowPower2020,ferrariEfficientNetworkFlooding2011,sahaEfficientManytoManyData2017}, and similar to the related  works~\cite{profentzasIoTLogBlockRecordingOffline2019,profentzasTinyEVMOffChainSmart2020,schillerBlockchainMSP430IEEE2020,novoScalableAccessManagement2019,abrahamSynchronousByzantineAgreement2019,abrahamSyncHotStuffSimple2020,shresthaOptimalityOptimisticResponsiveness2020,chaiByzantineFaultTolerant2014} consider network jamming to be an orthogonal problem.
Nevertheless, we note that the location of deployment such as data-centers, military networks using satellites, or remote large areas will dictate whether jamming is a part of the threat model.
If large-scale jamming is a part of the threat model for a setting, unicast links\footnote{Jamming in unicast links in wireless networks can be mitigated by Spread-Spectrum Frequency Hopping~\cite{torrieriPrinciplesSpreadSpectrumCommunication2018}.} are reliable.
The use of unicasts do not affect our claims on energy-efficiency over prior SMR protocols.

We measure and compare the costs of multicasts and unicasts and measure their reliability.
We implement our proposed protocol \protocol and the state-of-the-art protocol Sync HotStuff~\cite{abrahamSyncHotStuffSimple2020} on CPS devices and compare their energy costs.

\subsection{System setup}
We implemented our protocol and the baseline protocol Sync HotStuff in C++ and tested them on \texttt{NUCLEO-F401RE} nodes with ARM Cortex M4 $84$ MHz processors. The main board has $512$ kB of flash memory and $96$ kB of SRAM. The nodes communicate using Bluetooth Low Energy (BLE) modules. We used \texttt{Saleae Logic-Pro~8} and \texttt{INA169} sensors to measure the energy consumed by the nodes running the protocols.

\subsection{Evaluating primitives}\label{sec:exp:primitives}
We measure the energy costs of the communication and cryptographic primitives
described in \cref{sec:crypto primitives}.
\begin{table*}[htb]
\small
\caption{{\bfseries Sample energy consumption data for different media.} All
	the measurements are in milliJoule (mJ).}
	\begin{tabular}{l ccc cc cc}
		\toprule
		\multirow{2}{*}{\begin{tabular}[c]{@{}c@{}} Message \\
		Size\end{tabular}} & \multicolumn{3}{c}{BLE} &
		\multicolumn{2}{c}{4G LTE \cite{huangCloseExaminationPerformance2012}} & \multicolumn{2}{c}{WiFi}
		\\ \cmidrule(l){2-4} \cmidrule(l){5-6} \cmidrule(l){7-8}
        & Send & Recv & Multicast & Send & Recv & Send & Recv\\
        \midrule
		$256$ B & $0.73$ & $0.55$ & $0.58$ & $494.84$ & $69.54$
		& $81.2$ & $66.66$ \\
		$512$ B & $1.31$ & $1.11$ & $1.17$ & $989.68$ & $139.08$
		& $153.98$ & $123.23$ \\
		$1$  kB & $2.93$ & $2.64$ & $2.35$ & $1979.36$ & $278.17$
		& $310.54$ & $231.52$ \\
		$2$  kB & $5.91$ & $5.23$ & $4.70$ & $3958.72$ & $556.35$
        & $610.55$ & $423.58$ \\
        \bottomrule
	\end{tabular}
	\label{tab:energy_comm_primitives}
\end{table*}

\paragraph{Communication primitives}
We measured the cost of sending a message through different communication media in~\cref{tab:energy_comm_primitives}. The energy reading for BLE is the energy needed to transmit one packet of data. However, this BLE transmission does not verify if the other party received it reliably. We notice that BLE's energy requirement is two orders of magnitude lower than WiFi and three orders of magnitude lower than 4G, which is the primary reason of us using BLE in our implementation. We evaluated this mode of communication further by measuring the energy costs of sending and receiving messages of various sizes using unicasts and multicasts in BLE.
We use BLE advertisement packets as multicasts and \kcasts. We observe that these packets have payloads limited to $25$ bytes per packet by the BLE GAP specification. For large messages, we must therefore fragment them into multiple advertisements.
BLE advertisements are packets in the link layer, with no inherent capability to handle packet losses. We tackled this inherent unreliability using redundant transmissions. We conducted experiments to chart the reliability of \kcasts. The setup involved a set of BLE-enabled embedded devices placed within a $10$ meter distance of each other, which is within the range of BLE. We transmitted batches of $10,000$ packets with varied levels of redundancy and measured the number of successful packets received by the receivers. The failure rate was measured against the redundancy factor of transmission and the energy values were measured accordingly. A \kcast is considered to be a success only if all the \K receivers successfully receive the data. We repeated this experiment for different values of \K and recorded the experimental results in  \cref{fig:mc-f-l}. We observe that the failure rates exponentially decrease with the increase in the redundancy factor and energy spent per transmission. The probability of a transmission failure increases with the value of \K and the energy needed to guarantee a $99.99\%$ reliability also increases accordingly. However, for specific applications, a threshold of $99.99\%$ might not be sufficient. For such applications, the reliability can be customly set using concepts like erasure and error-correcting codes, where the amount of redundancy can be set according to the desired reliability. Based on the noise in the environment, we can measure a value $\delay$ such that the desired reliability is achieved. \footnote{It can be an interesting optimization to improve the reliability without redundant transmissions using erasure coding. However, we note that erasure coding solutions typically employ large pre-computed tables, which will not fit on the employed CPS boards without non-trivial implementation efforts.}

With the added redundancy, we experimentally observe that it takes
bounded $200$ ms to transmit a $25$ byte message with $99.99\%$ reliability over a
multicast link in BLE, with $\K=7$. BLE has a low bandwidth resulting in a high
transmission time, but it is energy efficient when compared to other
alternatives. The experiments were conducted in a laboratory with many rogue advertisement packets from other devices which increased the energy consumption of the receiver nodes. But this also mimics a practical setting due to the ubiquity of BLE.
For all further experiments, we use this setting of \kcasts that guarantees 99.99\% reliability for each k-cast link.

\begin{figure*}[t]
\begin{multicols}{3}
    \begin{subfigure}[t]{\linewidth}
    \includegraphics[width=\linewidth]{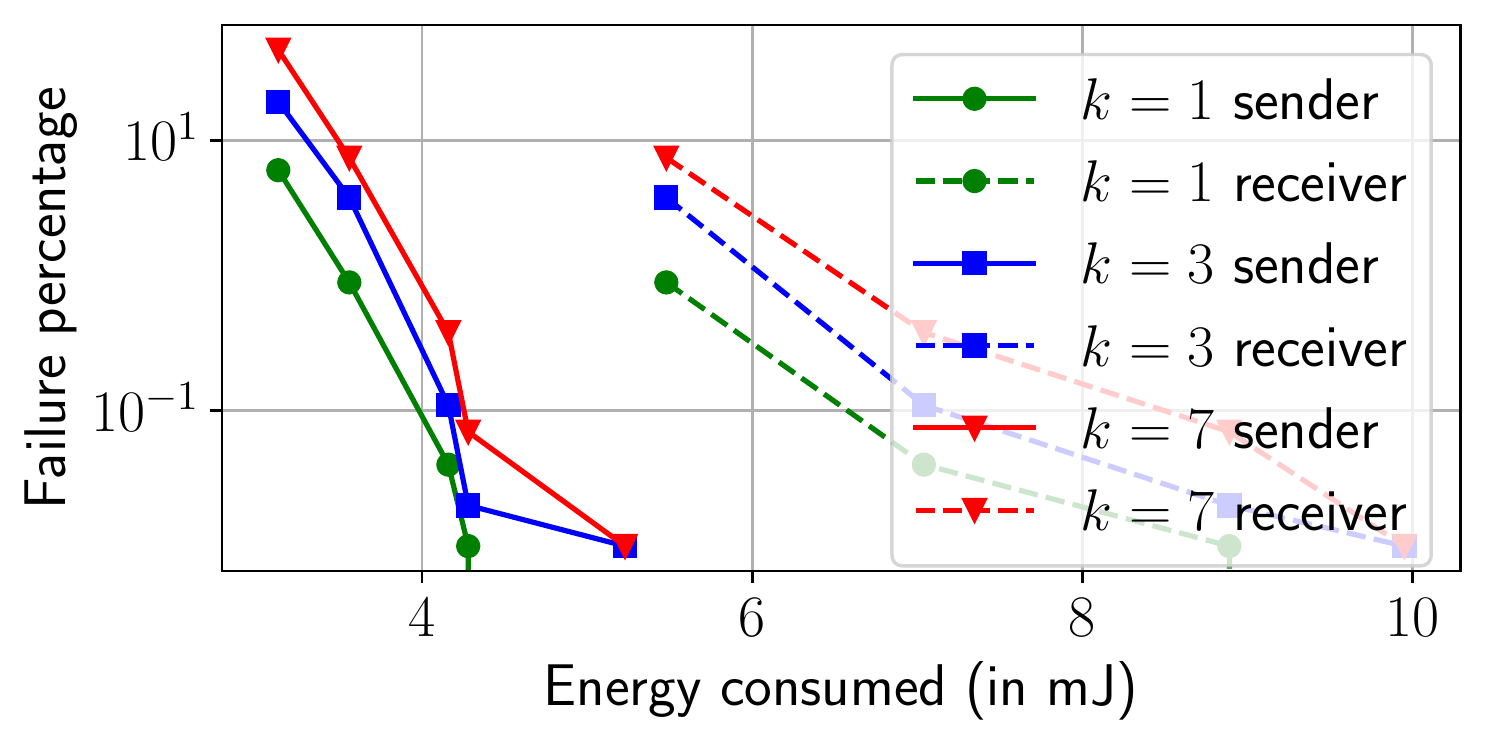}
    \caption{\ifsubmission\small\fi {\textbf{The failure rate of \kcasts} as a
    $\%$ mapped against the energy consumed by the sender and receiver. We achieve a reliability rate of $99.99\%$ by spending $5.3$ mJ
    per message by the sender and $9.98$ mJ per message for the receiver.}}
	\label{fig:mc-f-l}
    \end{subfigure}\par
    \begin{subfigure}[t]{\linewidth}
    \centering
    \includegraphics[width=\linewidth]{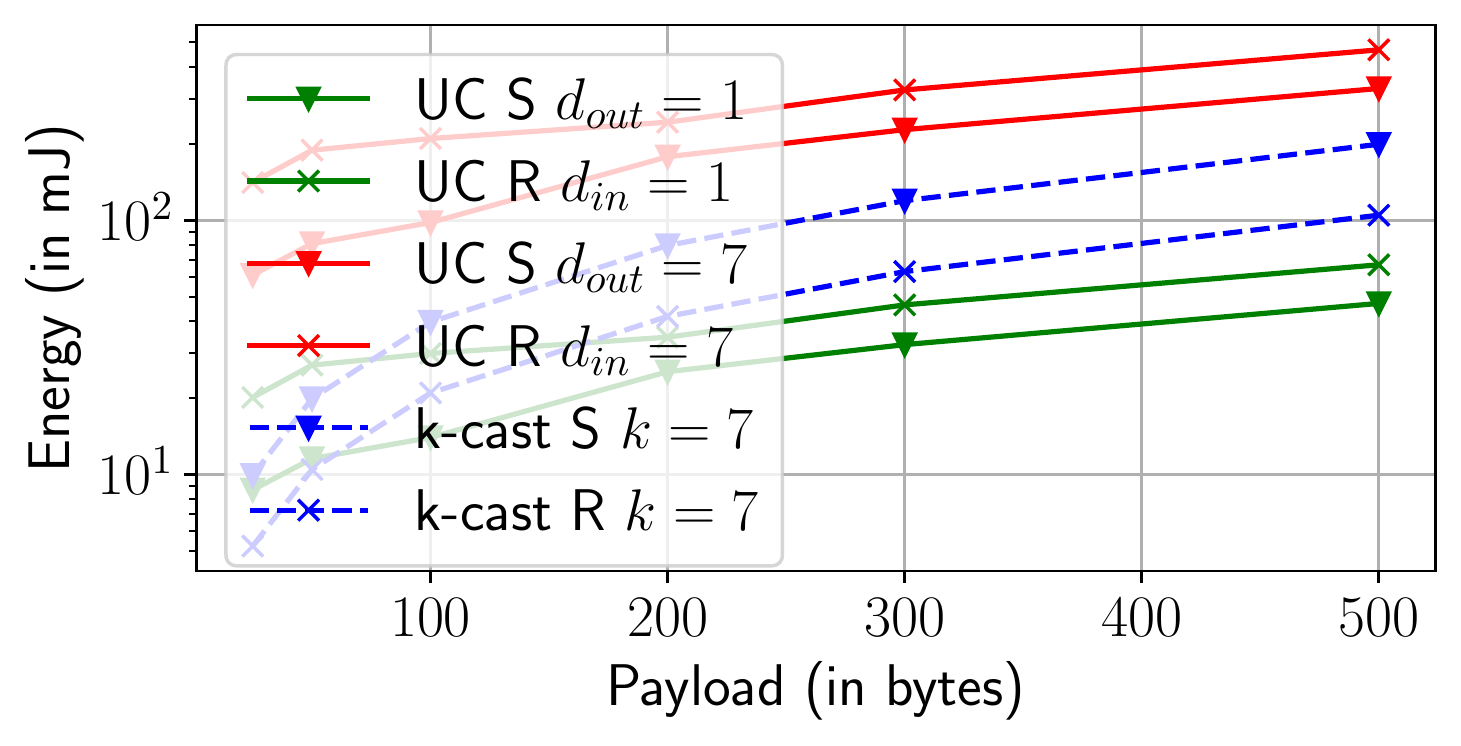}
    \caption{\ifsubmission\small\fi{\textbf{Comparing costs of unicast vs.
    multicasts}. We use a reliability of $99.99\%$ for \K-casts. UC stands for Unicast, S stands for Sender and R means receiver.}
        \label{fig: multicast vs. unicast}
    }
    \end{subfigure}\par
    \begin{subfigure}[t]{\linewidth}
    \includegraphics[width=\linewidth]{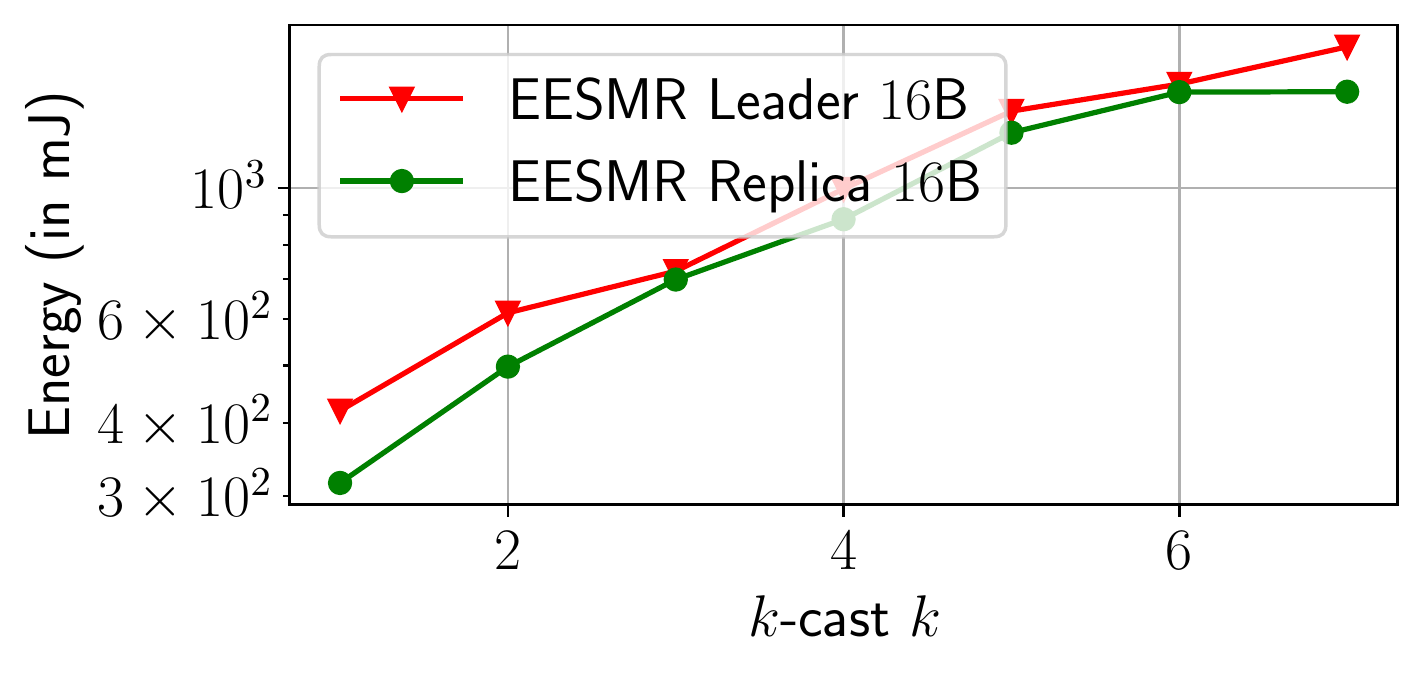}
    \caption{\ifsubmission\small\fi{\textbf{Average energy per state machine
    replication} consumed by a correct \protocol leader and other nodes and its
    variation with \K. We use $\sizeof{b_i} = 16$ bytes.}}
	\label{fig:k-e2c}
    \end{subfigure}\par
\end{multicols}
\vspace{-5mm}
\begin{multicols}{3}
    \begin{subfigure}[t]{\linewidth}
    \includegraphics[width=\linewidth]{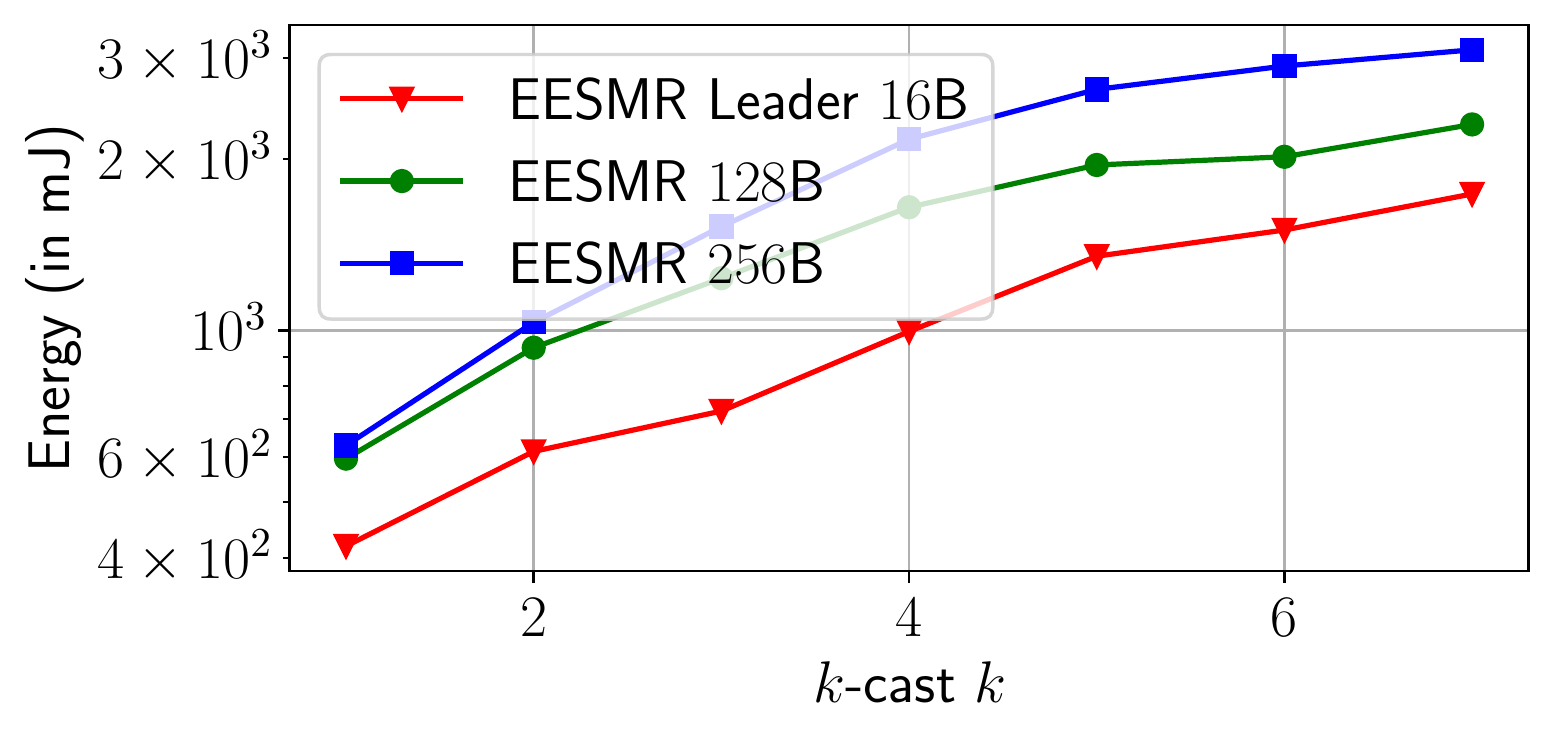}
    \caption{\ifsubmission\small\fi{\textbf{Energy consumed by a correct \protocol leader per SMR for variable block sizes} with respect to the value \K.}}
	\label{fig:e2c-bytewise}
    \end{subfigure}\par
    \begin{subfigure}[t]{\linewidth}
    \includegraphics[width=\linewidth]{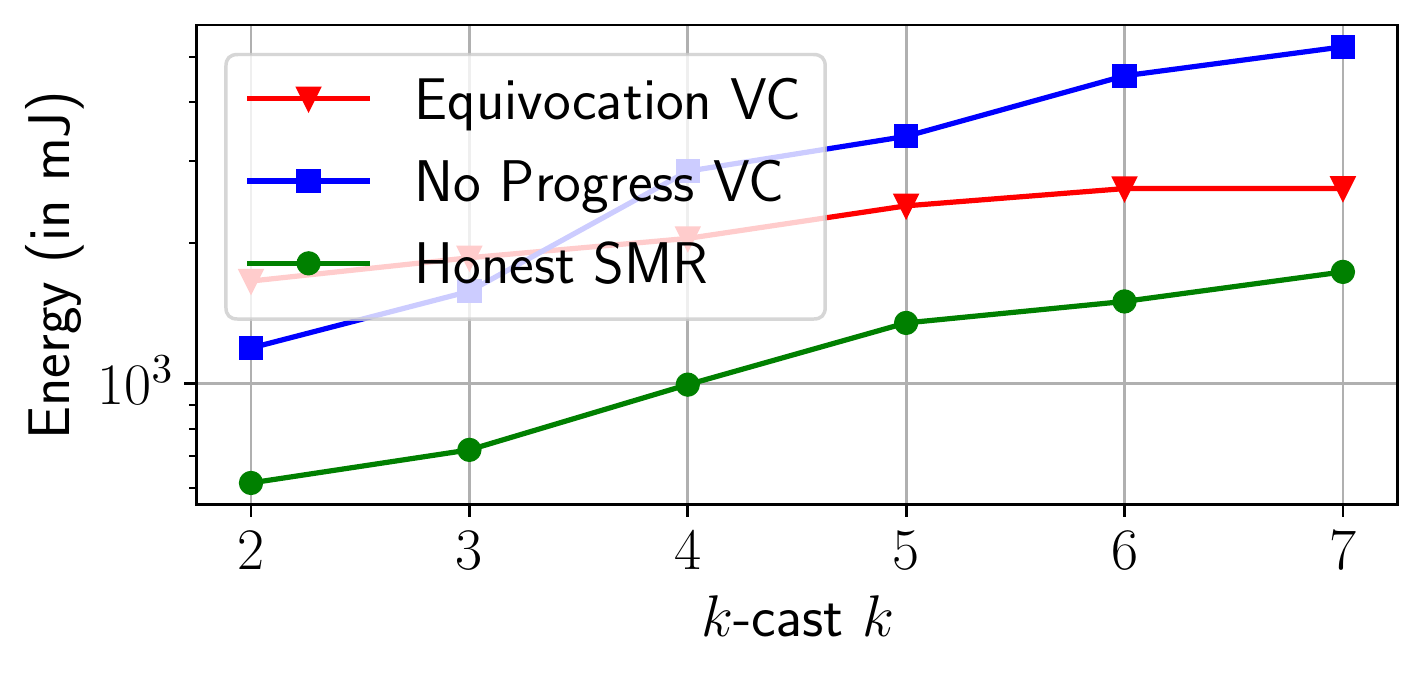}
    \caption{\ifsubmission\small\fi{\textbf{Energy consumed by \protocol leader per view change operation} for different $f$, evaluated with $n=15$ nodes. We use $\sizeof{b_i}= 16$ bytes.}}
	\label{fig:vc-e2c}
    \end{subfigure}\par
    \begin{subfigure}[t]{\linewidth}
    \includegraphics[width=1\linewidth]{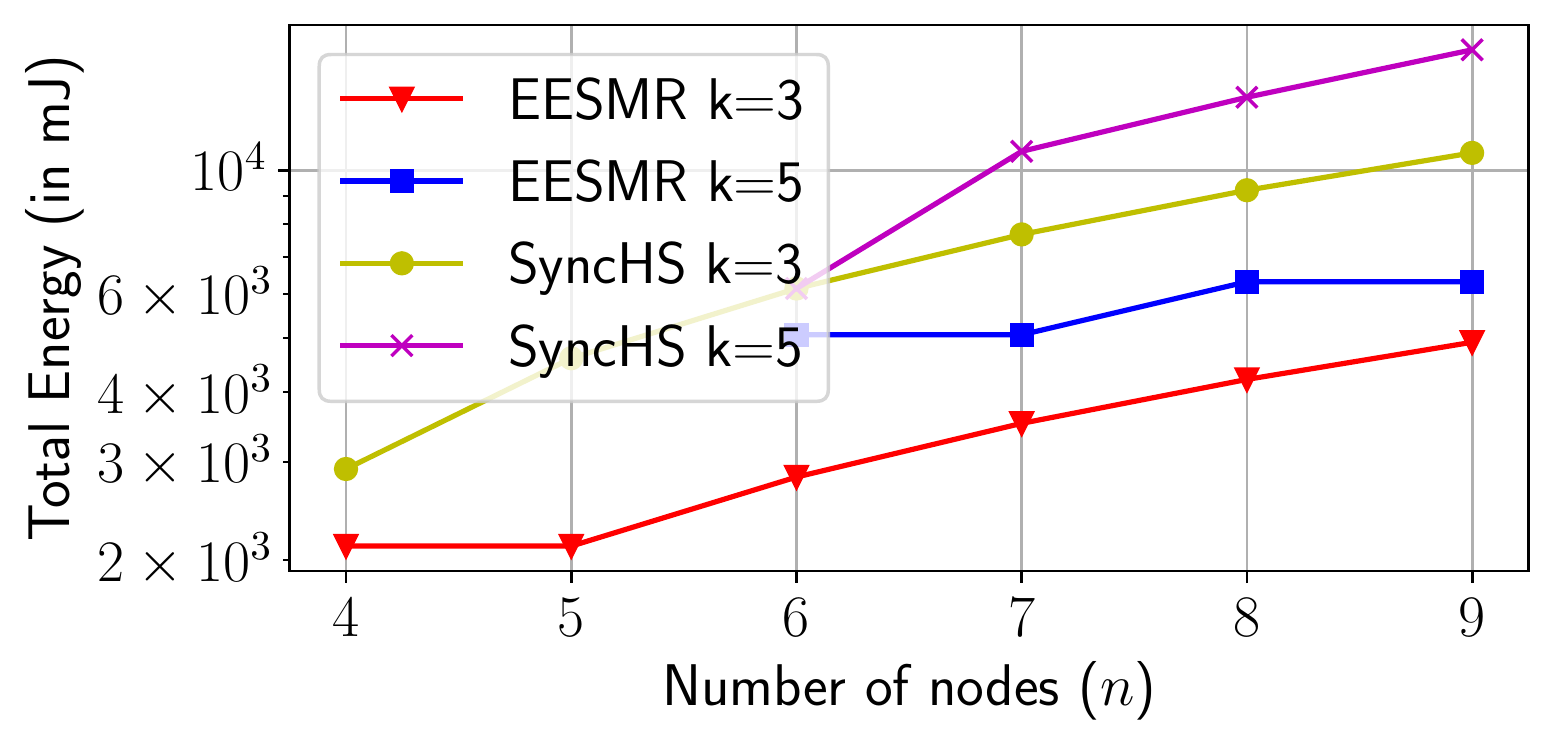}
    \caption{\ifsubmission\small\fi{\textbf{Total energy consumed by the correct nodes in \protocol and Sync HotStuff} with respect to the total number of nodes $n$.  }}
	\label{fig:E2CvsSyncHSTotalEnergy}
    \end{subfigure}
\end{multicols}        \vspace{-5mm}
\caption{\ifsubmission\small\fi\textbf{Energy characterization of various scenarios realized on embedded devices} }
\end{figure*}

\paragraph{Multicasts vs. unicasts on BLE}
We explore the unicast alternative to multicasts. Unicasts offers better
reliability at the cost of more point-to-point transmissions. The unicast links
use the GATT (\textbf{G}eneric \textbf{Att}ribute Profile) protocol, which is a connection-based protocol in the BLE realm. GATT offers reliable data transmission by inherently handling packet drops and re-transmitting data when necessary.   \cref{fig: multicast vs. unicast}
compares the energy consumption of $99.99\%$ \kcast on BLE with the equivalent
$\Dout = \K$ unicast links for BLE to transmit messages with different payloads.
We observe that expectedly, the energy required to transmit a \kcast using equivalent
unicasts increases linearly with the \K.

\cref{fig: multicast vs. unicast} compares the energy costs to transmit the same payloads using the $99.99\%$ reliable multicast and equivalent unicasts. We observe that a unicast link is more effective compared to a
multicast for bigger payloads, but this advantage is quickly negated as the
value of \K increases. The embedded devices used in our experiments
cannot handle concurrent unicast connections with its
neighbors, due to which data transfers using unicasts add extra time overheads.

\vspace{-1mm}
\subsection{Cryptographic Primitives}

\paragraph{Hash and HMAC costs}
We instantiate the MAC (message authentication code) algorithm using \texttt{SHA-256} as the underlying hash function, and measure the hashing cost for different message sizes.
We use short keys of recommended size $64$ bytes.
The cost of signing (\macsign) and verifying (\macvrfy) is the same as for the HMAC scheme.
The major cost in the HMAC scheme was mostly due to the underlying \texttt{SHA-256} algorithm.
We found that the cost of hashing increased linearly with message size.

\begin{table}[!htp]
\small
\centering
\caption{\ifsubmission\small\fi{\bfseries Comparing energy costs (in J) for signature generation
(\sigsign) and verification (\sigvrfy)} for ECDSA curves and RSA schemes.}
\begin{tabular}{l l cc}
\toprule
Algorithm & Parameters & Sign (in J) & Verify (in J)
\\ \midrule
\multirow{7}{*}{ECDSA} & \texttt{BP160R1} & $5.80$ & $11.03$ \\ %\cline{2-4}
& \texttt{BP256R1}   & $13.88$ & $27.34$ \\ %\cline{2-4}
& \texttt{SECP192R1} & $0.84$ & $1.50$ \\   %\cline{2-4}
& \texttt{SECP192K1} & $1.16$ & $2.24$ \\   %\cline{2-4}
& \texttt{SECP224R1} & $1.10$ & $2.14$ \\   %\cline{2-4}
& \texttt{SECP256R1} & $1.60$ & $3.04$ \\   %\cline{2-4}
& \texttt{SECP256K1} & $1.72$ & $3.35$ \\
\cmidrule(l){1-4}
\multirow{3}{*}{RSA} & $1024$-bit modulus & $0.40$ & $0.02$ \\ %\cline{2-4}
& $1260$-bit modulus & $0.79$ & $0.03$ \\ %\cline{2-4}
& $2048$-bit modulus & $2.41$ & $0.06$ \\ %\hline
\midrule
HMAC &
-- &
0.19 &
0.19
\\
\bottomrule
\end{tabular}
\label{tbl:pki_ops_energy}
\end{table}

\paragraph{Public key primitives}
We measured the energy costs for Elliptic Curve Digital Signature Algorithm (ECDSA)~\cite{johnsonEllipticCurveDigital2001} and RSA~\cite{cramerSignatureSchemesBased2000} for various security parameters, and recorded their energy readings in~\cref{tbl:pki_ops_energy}.
For ECDSA~\cite{johnsonEllipticCurveDigital2001}, we found that brainpool curves (\texttt{BP-XXX})~\cite{merkleEllipticCurveCryptography2010} were generally expensive to sign ($5$J), and verify ($11$J) for $160$ bit curves. NIST optimized curves (\texttt{SECP-XXX})~\cite{fips186DigitalSignature1994} offer better performance in comparison ($1$J and $2$J respectively) for $160$-bit curves. We use the implementations from MbedTLS~\cite{READMEMbedTLS2022} (commit: \texttt{b6229e304}) for our measurements. We implemented \text{BP160R1} using MbedTLS and RFC 5639~\cite{merkleEllipticCurveCryptography2010}. Finally, RSA using $1024$-bit modulus costs $400$mJ and $20$mJ to sign and verify, being an ideal candidate for CPS.

RSA using $1024$-bit modulus provides $80$-bits of security and should be practical for most CPS settings~\cite{kleinjungFactorization768BitRSA, IBMDocumentation2021}.
RSA provides two benefits over ECDSA: (i) reduced energy costs, and (ii) the benefit of asymmetry in verification and signing costs.
The latter implies higher energy costs for Byzantine nodes (to equivocate) while the impact is significantly less on the correct nodes that verify.

\iffull
\paragraph{A remark on hash based signatures}
Hash based signature schemes are claimed to be energy efficient but the public key sizes are large or the signature sizes are large.
For instance, signature schemes like Winternitz-OTS~\cite{merkleCertifiedDigitalSignature1990, lafranceSecurityWOTSPRFSignature2019} and HORS~\cite{reyzinBetterBiBaShort2002} have large key sizes for one time use.
This results in large storage requirements for CPS nodes. Multi-time hash based signatures such as XMSS~\cite{buchmannXMSSPracticalForward2011} and HORST~\cite{bernsteinSPHINCSPracticalStateless2015} require large intermediate memory leading to large RAM requirements (larger than what can be found in our CPS nodes).
For example, consider XMSS with the compression parameter $w = 2^8$ and re-usable for $2^{10}$ signatures.
For messages of size $32$ bytes (output of \texttt{SHA-256} for example), the signature is $102$ bytes long with large intermediate memory requirements.
It requires $765$ hash computations to sign a message, and $785$ hash computations to verify.
Due to large memory requirements for hash based signature schemes, we do not consider them in this work.
\fi

\subsection{Evaluating \protocol}\label{sec:experiments}

We implemented \protocol on embedded devices using BLE advertisements as \kcasts with
sufficient redundancy to achieve a reliability of $99.99\%$ for message transmission. Considering the memory limitations of the devices we used, we performed experiments on the blocking variant of \protocol.\footnote{The non-blocking variant of our protocol requires more memory (theoretically unbounded) to process all the proposed blocks. The only changes required in the protocol are to let the timers run concurrently, and multiple proposals can be forwarded by a node. The energy analysis still hold with respect to every block. A common issue with non-blocking protocols (all works so far) is that an adversary can propose a large number of non-blocking proposals but revert them during a view-change leading to wasted energy. We leave this as an orthogonal problem of independent research, as it is generally applicable to all non-blocking protocols. This can be mitigated by bounding the number of allowed non-blocking proposals.}
We use $\Din=\K$, and $\Dout=1$ in our
experiments, where each hyper-edge has degree \K. The network topology is defined as follows: In a system of \parseVn{n}
nodes using \kcasts, every node $\nodei{i}$ transmits messages to nodes $\nodei{i+1 \bmod n}\ldots,\nodei{i+k \bmod n}$. Every node $\nodei{i}$ receives messages from $\nodei{i-1 \bmod n}, \ldots \nodei{i-k \bmod n}$, because of which only these nodes need to be in communication range. The physical placement of the nodes is such that each node is within communication range of all nodes with which it has an edge. We measured the energy consumption of \protocol by varying $n$, \K and the block size. We take the \delay value to be $10n$ seconds. We took this metric because of the inability of our devices to scan and transmit data simultaneously. We designed a time-sequenced schedule and allocated an interval of $10$ seconds for every node to transmit data reliably. In this interval, all other nodes will only listen to packets from the designated node scheduled for transmission. Even if Byzantine nodes try and transmit their messages in undesignated intervals, the honest nodes will ignore their packets and only listen to the designated node. We show that the network is synchronous with this value of \delay. However, in rare situations where the network's \delay is higher, honest nodes can trigger the view change protocol to replace the leader, even when the leader was honest. We
instantiate the abstract block structure, concretely as $B = \msg{m,H\p{\defaultblkcontent},H\p{\defaulthashptr},\msg{i,H(b_i)}_L}$,
where $H$ is a cryptographic hash function, $m$ is the block height, and $L$ is the
leader.
In CPS, these blocks contain data \defaultblkcontent to maintain state.

We also implemented optimizations for the equivocation scenario, by avoiding construction of quorum certificates and using the equivocation and proof to quit the view.
We also optimize the no-progress scenario for the blocking version of our protocol, by avoiding the highest certificate construction, and ensuring that the status message only contains node's locked block.
This is secure because the highest committed block is either the locked block or its parent, and thus it is safe to extend any one of the locked blocks among $f+1$ locked blocks.

We recorded the energy consumed by the protocol by subtracting the
energy consumed by the device in sleep state. While doing so, we noticed that each node consumes 0.3 mW of power while in the sleep state and consumes approximately 1 mW of power while conducting SMR. Even at a modest frequency of one SMR per hour, the energy consumed by the SMR causes a substantial load on the node's limited energy resources. Hence, energy-optimal SMR protocols are necessary to reduce the overhead of SMR on lightweight nodes.

Our first observation is that the energy cost of \protocol is
independent of $n$ in the best case as we do not use certificates (which are
$f+1$ signatures). The energy cost only depends on \K. We show this in \cref{fig:E2CvsSyncHSTotalEnergy}, where the total energy consumption of all the correct nodes increases linearly with respect to the number of correct nodes in the network.

Next, we measured the difference in energy costs between the leader and the other
correct nodes for different values of \K. In \cref{fig:k-e2c}, we observed this difference. The linearity of energy consumption with \K is due to \K
incoming edges. The two energy costs are close, with the leader having a higher cost. We also evaluated the variation of energy consumption with block size and show that \protocol scales well with increasing message payloads (\cref{fig:e2c-bytewise}).

We measured the energy consumed by \protocol to perform a view change operation, in case of a Byzantine leader. The energy consumed per view change is plotted % (Fig. 15)
(\cref{fig:vc-e2c})
with respect to the value of $f$, in case of an equivocating leader and a stalling leader. The value of \K has been taken to be $f+1$ to measure the minimum energy that a leader spends in the view change. Comparison with correct case SMR has also been performed with \K=$f+1$. The stalling case is much more expensive in terms of energy because of the vote casting and the $f+1$ blame messages that need to be verified. It should be noted that the energy measurement for the view change is only for the blocking version of \protocol{}.
\begin{figure}[!ht]
    \centering
    \includegraphics[width=\linewidth]{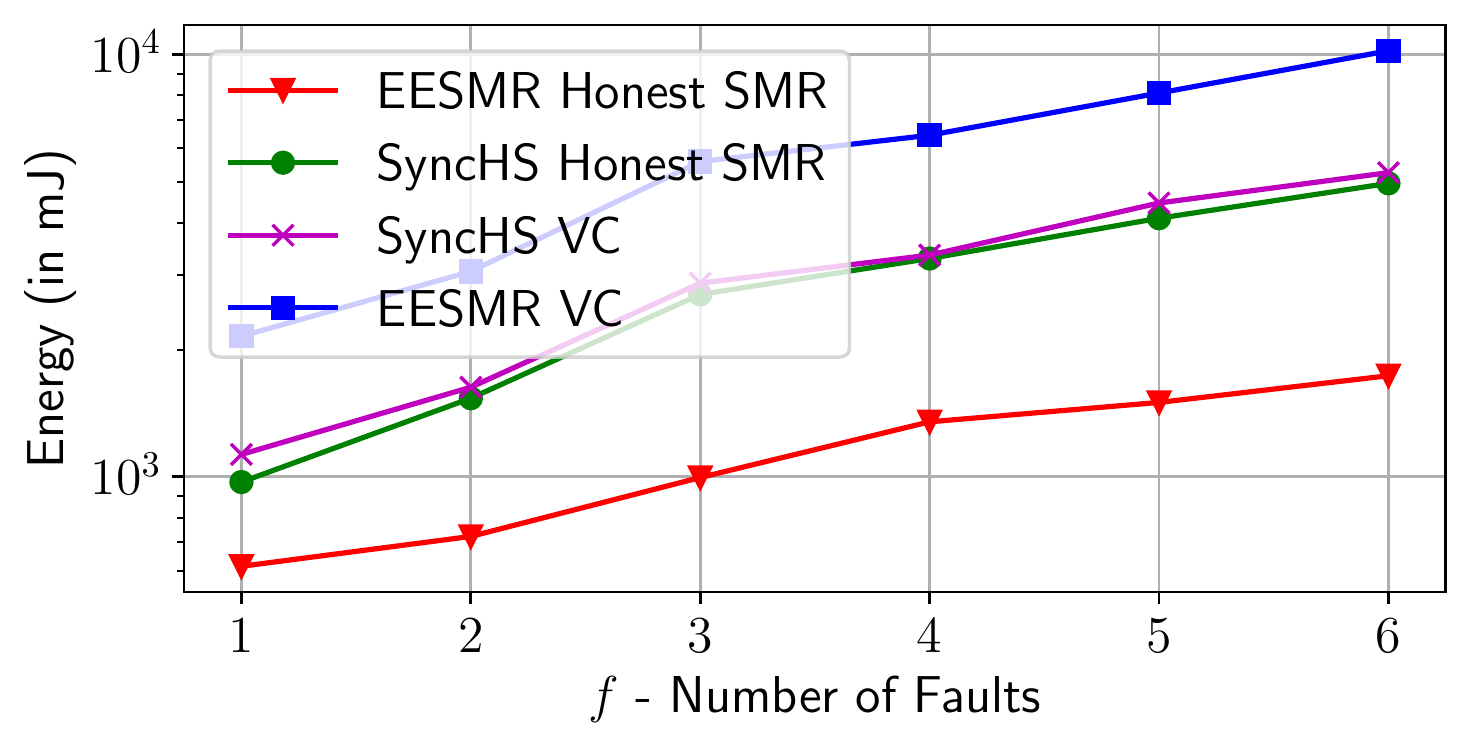}
    \vspace{-3mm}
    \caption{\ifsubmission\small\fi{\textbf{The energy consumed by a leader in \protocol vs. Sync HotStuff} to tolerate $f$ Byzantine faults in a system of $n=13$ nodes. We show both the faulty and non-faulty leader cases.}}\label{fig:e2cvssynchs}
    \ifsubmission{}\Description{}\fi
\end{figure}

\begin{table*}[!htbp]
    \renewcommand{\thefootnote}{\fnsymbol{footnote}}
	\small
	\caption{\ifsubmission\small\fi{{\bfseries Comparison of best-case and worst-case scenarios in
		  related SMR works.}
		Here, the number of nodes is $n$;
		We consider a partially connected network with all nodes connected to $d$ neighbors such that the graph remains connected even if $f$ nodes are Byzantine.
		Some of the protocols are for fully connected graphs but we compare them in our partially connected setting.
		\emph{Block period} refers to the time difference between two successive proposed blocks.
		$\delta \le \delay$ represents the actual message delivery time.}}\label{tab:comp other work SMR}
	\begin{tabular}{l cccc cccc}
		\toprule
		\multirow{3}{*}{Protocol} & 
		\multicolumn{4}{c}{Correct Leader (best-case)} &
		\multicolumn{4}{c}{Faulty Leader (worst-case)} 
		\\ \cmidrule(l){2-5} \cmidrule(l){6-9}
		& 
		\multirow{2}{*}{
			$
			\begin{matrix}
			\text{Communication} \\
			\text{Complexity}
			\end{matrix}
			$
		} &
		\multicolumn{2}{c}{Public Key Ops.} &
		\multirow{2}{*}{
			$
			\begin{matrix}
				\text{Block} \\
				\text{Period}
			\end{matrix}
			$
		} &
		\multirow{2}{*}{
			$
			\begin{matrix}
			\text{Communication} \\
			\text{Complexity}
			\end{matrix}
			$
		} &
		\multicolumn{2}{c}{Public Key Ops.} &
		\multirow{2}{*}{
			$
			\begin{matrix}
				\text{Block} \\
				\text{Period}
			\end{matrix}
			$
		}
		\\ \cmidrule(l){3-4} \cmidrule(l){7-8}
		& 
		&
		Sign &
		Verify &
		&
		&
		Sign &
		Verify 
		\\ \midrule
%  		Dolev \& Strong~\cite{dolevAuthenticatedAlgorithmsByzantine1983} & 
%  		$O(n^2d)$ & 
%  		$O(n)$ &
%  		$O(n^2)$ &
%  		$O(n^2d)$ & 
%  		$O(n)$ &
%  		$O(n^2)$ 
%  		\\
		Abraham \etal \cite{abrahamSynchronousByzantineAgreement2019} & 
		$O(n^2d)$ & 
		$O(n)$ & 
		$O(n^2)$ &
		--- &
		$O(n^3d)$\footnotemark[1] &
		$O(n)$ &
		$O(n^2)$ &
		---
		\\
		Sync HotStuff \cite{abrahamSyncHotStuffSimple2020} & 
		$O(n^2d)$ & 
		$O(n)$ & 
        $O(n^2)$ &
		$2\delta$ & 
        $O(n^3d)$\footnotemark[1] &
        $O(n)$ &
		$O(n^2)$ &
		$14\delay$
		\\
		OptSync \cite{shresthaOptimalityOptimisticResponsiveness2020} & 
		$O(n^2d)$ & 
		$O(n)$ & 
        $O(n^2)$ &
		$2\delta$ &
        $O(n^3d)$\footnotemark[1] &
        $O(n)$ &
		$O(n^2)$ &
		$14\delay$
		\\ 
		Rotating BFT SMR~\cite{abrahamOptimalGoodCaseLatency2022} &
		$O(n^2d)$ &
		$O(n)$ &
		$O(n^2)$ &
		$2\delta$ &
		$O(n^2d)$\footnotemark[2] &
		$O(n)$ & 
		$O(n^2)$ &
		$14\delay$
		\\
		\midrule
% 		\ifsubmission\else
        \textbf{\protocol} & 
 		$\mathbf{O(nd)}$ & 
 		$\mathbf{O(1)}$ & 
 		$\mathbf{O(n)}$ &
		$\mathbf{0}$ &
 		$\mathbf{O(n^3d)}^{*}$ &
        $\mathbf{O(n)}$ &
		$\mathbf{O(n^2)}$ &
		$\mathbf{21\delay}$
		\\ 
		\bottomrule
	\end{tabular}
	
{\footnotemark[1]A Byzantine leader can make the communication complexity in a view change to be $O(n^3d)$, but this results in all correct nodes committing $O(n)$ blocks after the view change. 
\footnotemark[2]For the rotating leader protocol, we report the worst-case for one iteration with a bad leader. 
In theory, we could have $o(n)$ bad leaders consecutively, leading to an increase in all the numbers by a factor of $n$. 
}
\end{table*}

%%% Local Variables:
%%% mode: latex
%%% TeX-master: "main"
%%% End:

\vspace{-4mm}
\subsection{Comparison with Sync HotStuff}
We compare in~\cref{fig:e2cvssynchs} \protocol's energy readings with our implementation of Sync HotStuff~\cite{abrahamSynchronousByzantineAgreement2019}. We made simplifying assumptions in favor of Sync HotStuff, by partially implementing vote forwarding. The minimum energy to tolerate $f$ faults has been measured for \protocol and Sync HotStuff.
We designed our $n=13$ graph to ensure that it is $f$ fault-tolerant. We measure the energy consumed by a correct leader in a network consisting of $\lfloor{\frac{n}{2}}\rfloor$ faulty nodes.
The amount of energy consumed by a node linearly increases with the value of \K. As \K increases, every node receives messages from an increasing number of nodes, which increases the energy spent in receiving them. Additionally, Sync HotStuff uses a certificate of size $f+1$ increasing the energy costs. When the leader is correct, Sync HotStuff is $2.85\times$ more energy hungry than \protocol. In a view change, the ratio of the energy consumed by \protocol Energy with respect to Sync HotStuff is $2.05$. To repeat our earlier argument, in the steady state, the number of SMR rounds where there will be a view change will become far smaller than where the leader will be benign. We also compare the total energy consumed by correct nodes per SMR for both protocols in  \cref{fig:E2CvsSyncHSTotalEnergy}. This demonstrates the scalability of \protocol in the best-case.

%%% Local Variables:
%%% mode: latex
%%% TeX-master: "main"
%%% End:

% Related Works
\section{Related Work}\label{sec:multi_related}

To the best of our knowledge, there is no prior work that formally analyzes SMR protocols for energy-efficiency.

\paragraph{SMR}
Prior works~\cite{yinHotStuffBFTConsensus2019,buchmanLatestGossipBFT2019,abrahamSyncHotStuffSimple2020, castroPracticalByzantineFault2002a,golanguetaSBFTScalableDecentralized2019,bessaniStateMachineReplication2014,kotlaZyzzyvaSpeculativeByzantine2007a,abrahamOptimalGoodCaseLatency2022} focused on improving the message complexity or the commit latencies of SMR.\@ 
But nodes in CPS tend to be constrained in their energy and computational resources, and decreasing communication complexity and increasing the computation does not help the CPS nodes. Protocol designs as of today use votes and cryptographic certificates while completely ignoring their energy costs on CPS devices. In the absence of signature aggregation techniques, these protocols~\cite{golanguetaSBFTScalableDecentralized2019,yinHotStuffBFTConsensus2019,abrahamSyncHotStuffSimple2020} increase the energy requirements for SMR.\@ 
We present a comparison of the related works in \cref{tab:comp other work SMR}.

Let $\delta$ be the actual network speed between any two nodes in a fully connected system.
OptSync~\cite{shresthaOptimalityOptimisticResponsiveness2020} commits and produces a block every $2\delta$. 
To achieve this all nodes verify $3n/4+1$ signatures. 
Sync HotStuff~\cite{abrahamSyncHotStuffSimple2020} on the other hand produces a block every $2\delta$, has a commit latency of $2\delay$, and verifies $n/2+1$ signatures. 
Clearly, Sync HotStuff is more energy efficient than OptSync with an increase in latency from $2\delta$ to $2\delay$. 
\protocol can create blocks as soon as it has enough messages without any other delays giving it a block period of $0$.
It has a commit latency of $4\delay$ (accounting for chain-synchronization), and uses only $1$ signature. 
Thus, \protocol{} is more efficient than both Sync HotStuff and OptSync to commit blocks.

\paragraph{Efficient SMR} 
MinBFT~\cite{veroneseEfficientByzantineFaultTolerance2013} and CheapBFT~\cite{kapitzaCheapBFTResourceefficientByzantine2012} are partially synchronous SMR protocols that use trusted counters and subsystems to improve the amount of communication of PBFT~\cite{castroPracticalByzantineFault2002a}.
They use trusted components to assign sequence numbers to requests, thereby restricting the equivocating capabilities of Byzantine adversaries. They also use only $f+1$ \emph{active} replicas which replicate the state thereby saving system resources. Apart from the stronger assumptions, they still use certificates per block and $3$ rounds of all-to-all diffusion, leading to the similar energy inefficiencies.

Other metrics of efficiency such as bit-complexity~\cite{duanByzantineReliableBroadcast2022,dasAsynchronousDataDissemination2021,alhaddadSuccinctErasureCoding2021,nayakImprovedExtensionProtocols2020} or latency~\cite{abrahamGoodcaseLatencyByzantine2021,abrahamByzantineAgreementOptimal2015,abrahamOptimalGoodCaseLatency2022} have been considered in the literature.
However, they use impractical techniques such as expander graphs or computationally expensive primitives such as threshold cryptography or online error correction.

There is an orthogonal set of works that improve the efficiency of protocols by using trusted components~\cite{yandamuriCommunicationEfficientBFTUsing2023,decouchantDAMYSUSStreamlinedBFT2022}.
Indeed, as evidenced in \cref{fig:e2c feasible region}, trusted components can significantly lower energy costs.
But existing works assume all-to-all communication which makes the protocol flood $O(n)$ messages in every round.
However, it remains an interesting open problem to combine our techniques with trusted components for improved energy-efficiency.

\iffull{}
\paragraph{BA and SMR}
Byzantine agreement (BA)~\cite{dolevAuthenticatedAlgorithmsByzantine1983} is an agreement problem
that is closely related to SMR, where all the correct nodes need to agree on a
message sent by the sender in the presence of malicious nodes. Informally, a BA
protocol satisfies \textit{safety} which requires all correct nodes to commit to
the same value and \textit{validity} which requires that if all correct nodes
start with value $v$ then all correct nodes must output $v$.

A key difference between BA and SMR that we consider in this work, is the
\textit{validity} condition. In BA, if all correct nodes start with a value $v$,
then the output of the protocol for all correct nodes must be $v$. However, an
SMR protocol abstracts the value $v$ as requests from clients, and leaves the
validity to the semantics or application layer and not the consensus layer.
\fi

\paragraph{Using multicasts} Multicasts occur frequently due to the spatial arrangement of CPS nodes and the omni-directional nature of wireless communication media. It may not cost more energy to use multicast links over unicast links. Several prior works~\cite{kooBroadcastRadioNetworks2004, montemanniMinimumPowerBroadcast2005, guoEnergyawareMulticastingWireless2007, fitziPartialConsistencyGlobal2000, considineByzantineAgreementGiven2005, khanExactByzantineConsensus2019} leverage multicasts to address BA. A common drawback of these works is that they are not tailored to leverage practical multicasts, e.g., considering the drop in reliability with increasing degree.

Koo \emph{et al.}~\cite{kooBroadcastRadioNetworks2004} provide a bound on the number of
Byzantine nodes that can be tolerated in the neighborhood of a correct sender
while using multicasts to achieve BA.\@ Montemanni \emph{et al.}
~\cite{montemanniMinimumPowerBroadcast2005} and Guo \emph{et al.} ~\cite{guoEnergyawareMulticastingWireless2007}
propose models in determining the optimal transmission power of a sender. However, these works~\cite{kooBroadcastRadioNetworks2004, montemanniMinimumPowerBroadcast2005, guoEnergyawareMulticastingWireless2007} cannot be directly used for CPS since the transmission power of the antennae is typically fixed.

Fitzi \etal~\cite{fitziPartialConsistencyGlobal2000} assume the existence of reliable multicasts
between \textit{every} three nodes. They show that the bound on the number of \faulty nodes that can be supported in such a fully connected synchronous network can be improved from $n/3$ to $n/2$ to solve BA. Considine \emph{et al.}~\cite{considineByzantineAgreementGiven2005} show that a resilience of $\tfrac{n}{f} > (k+1)/(k-1)$ can be achieved in a system with $n$ nodes and \kcasts, with $f$ of them being \faulty. Khan \emph{et al.}~\cite{khanExactByzantineConsensus2019} consider multicasts as undirected edges in a graph and provide necessary and sufficient conditions for BA in their model. The main drawback in these works is that there is an assumption that \emph{every} subset of $k$ nodes have a reliable \kcast present. In this work, we make weaker multicast assumptions of links existing only between neighboring nodes.

%%% Local Variables:
%%% mode: latex
%%% TeX-master: "main"
%%% End:

% Discussion
\section*{Conclusion}

We present \protocol, an SMR protocol with improved communication in partially connected networks, signature generation, and verification costs.
We present analysis techniques to determine choice of protocols and argue why optimizing the best-case is important.
We provide a general hypergraph model that can take advantage of multicasts in wireless CPS, when available.  
Finally, we empirically show a $33-64\%$ reduction in energy costs in the steady-state as compared to the state-of-the-art solution Sync HotStuff.

% Acknowledgement
%%% Local Variables:
%%% mode: latex
%%% TeX-master: "main"
%%% End:

\section{Acknowledgments}

We would like to thank the various reviewers and our shepherd Jérémie Decouchant for the helpful feedback and suggestions to improve this draft significantly.
We would like to thank Samarjit Chakraborty, Kartik Nayak and Nibesh Shrestha for helpful feedback about our draft.
This work was supported in part by NIFA award number 2021-67021-34252,the National Science Foundation (NSF) under grant CNS1846316 and National Science Foundation Cyber Physical Systems (CPS), the United States Department of Agriculture, and the Army Research Lab Contract number W911NF-2020-221.

\balance{}
\bibliographystyle{ACM-Reference-Format}
\bibliography{References-core,extra}

\iffull{}
\appendix
\section{Fault Tolerance in Hypergraphs}\label{sec: fault tolerance result}

\subsection{Preliminaries}
\paragraph{Hypergraph} To incorporate the multicasts available in the CPS settings, we model the network as a strongly connected hypergraph \h.
Formally, we define a hypergraph as follows:

\begin{definition}[Hypergraphs]\label{def:hypergraph}
\parsehgraph is a hypergraph, where \parseV is the set of nodes, and the hyper-edge set $\E \subseteq \V \times 2^{\V}$, where $2^{\V}$  is the power set of \V excluding the set $\V^{0}$, i.e., the empty set.
\end{definition}

In a hypergraph \parsehgraph, the edge set \E models the multicasts between the nodes.
Consider an example where $\nodei{1}$ has a multicast with $\nodei{2}$, $\nodei{3}$, and $\nodei{4}$.
This can be represented as $e = \hedge{1}{\nodei{2},\nodei{3},\nodei{4}} \in \E$.
In general, for an edge $e$, $S(e)$ denotes the sender, and $R(e)$ denotes the set of receivers of $e$.
Our definition also excludes self-loops, i.e., $\forall e \in \E$, $S(e) \not\in R(e)$.

\gdef\Din{\ensuremath{D_{in}}\xspace}
\gdef\Dout{\ensuremath{D_{out}}\xspace}

\emph{\kcasts, \Din, and \Dout.}
We say our hypergraph \h has \kcasts if every edge contains at least $\K$ receivers.
\Din denotes the minimum number of edges in which a node is a receiver, among all the nodes.
Similarly, \Dout denotes the minimum number of edges in which a node is a sender, among all the nodes.

\paragraph{Independence of edges}
In the unicast setting, every edge enables communication with new nodes.
This is not always true for hypergraphs.
Consider three edges $e_{1}$, $e_{2}$, and $e_{3}$ with the same sender $S(e_{\cdot})=\nodei{i}$ having $R(e_1) = \set{\nodei{1},\nodei{2}}$, $R(e_2) = \set{\nodei{2},\nodei{3}}$ and $R(e_3) = \set{\nodei{1},\nodei{3}}$.
All three edges are unique but one of them is redundant and can be removed while still reaching the nodes $\nodei{1}$, $\nodei{2}$, and $\nodei{3}$.
We call hypergraphs without such redundant edges as having an \textit{independent set of edges} (\cref{def:ind edges}).
In our discussion, we assume that all hypergraphs have independent set of edges.
A modified spanning tree algorithm can be used to find such sets.

\begin{definition}[Independence of edges]\label{def:ind edges} 
	Edges \E in a hypergraph \parsehgraph are said to \emph{independent} if for every node $\nodei{i}$ there does not exist two distinct subsets $\E_1,\E_2 \subseteq \E$ with $S(e_{i})=\nodei{i}$ for every $e_{i}\in\E_{1}\cup\E_{2}$  such that $\bigcup\limits_{e_i \in \E_1} R(e_i) = \bigcup\limits_{e_j \in \E_2} R(e_j)$.
\end{definition}

\paragraph{In- and out-degrees}
In a hypergraph, a node $\nodei{i}$ has two degrees: the in-degree $\din(\nodei{i})$ and the out-degree $\dout(\nodei{i})$.
We use the global minimum as the hypergraph's in-degree and out-degree.
Formally, it is defined as follows:
\begin{definition}[In-degree \din]
  For any node $\node_i \in \V$ of a hypergraph \parsehgraph, the in-degree of node $\nodei{i}$, denoted by $\din\p{\nodei{i}}$, is the number of unique nodes $\nodei{j} \in \V$, such that there exists hyper-edges $e\in\E$, in which $S(e)=\nodei{j}$ and $\nodei{i}\in R(e)$.
  We denote by \din, the minimum $\din\p{\nodei{i}}$ among all nodes $\nodei{i} \in \V$.
\end{definition}
\begin{definition}[Out-degree \dout]
  For any node $\nodei{i} \in \V$ of a hypergraph \parsehgraph, the out-degree of node $\nodei{i}$, denoted by $\dout\p{\nodei{i}}$, is the number of unique nodes $\nodei{j}\in\V$, such that there exists hyper-edges $e\in \E$, in which $S(e)=\nodei{i}$ and $\nodei{j}\in R(e)$.
  We denote by \dout, the minimum $\dout\p{\nodei{i}}$ among all nodes $\nodei{i} \in \V$.
\end{definition}

\paragraph{Network delay} 
In this section, we overload \delay as the \textit{maximum network delay} over the partially connected hypergraph \h.
The \delay parameter is an upper bound on the time a message from any correct sender takes to reach any other correct node in the network, possibly after the flooding.
\delay captures delays on each edge, the computations performed, re-transmissions due to collisions/link failures, verifying the integrity of messages and other scheduling factors involved in the network.

For the clarity of presentation, we assume that all nodes have clocks with zero drift and their clocks are within $\delay$ of each other which can be achieved using clock synchronization protocol~\cite{abrahamSynchronousByzantineAgreement2019}. 
A bounded clock drift can be handled by running clock synchronization protocols regularly~\cite{abrahamSynchronousByzantineAgreement2019}, to ensure that the drift does not grow unbounded.

\subsection{Hypergraph Fault-tolerance}

For any network graph, we need to ensure that Byzantine nodes cannot partition the correct nodes.
Even if one correct node \nodei{i} becomes unreachable because all its neighbors do not forward its messages, there is one less correct node in the protocol than $n-f$, and the protocol is not secure.

For an undirected partially connected network, if $d$ is the \textit{minimum} degree among all the nodes in the graph, the necessary condition to prevent partitioning is $f < d$~\cite{dolevAuthenticatedAlgorithmsByzantine1983}.
For directed partially connected network graphs, if $d_{i}$ and $d_{o}$ are the minimum in-degree and out-degree among all the nodes in the graph, then the necessary condition is $f < \min\p{d_{i},d_{o}}$.
Intuitively, this generalizes that all correct nodes must be able to send messages to all other correct nodes as well as be able to receive messages from all other correct nodes.

In our hypergraph model \parsehgraph with \kcasts, $\dout\p{\nodei{i}}$ is the number of distinct nodes to which the node $\nodei{i}$ can send messages, and $\din\p{\nodei{i}}$ be the number of distinct nodes from which $\nodei{i}$ can receive messages.
For fault-tolerance, this implies $f < \min\limits_{\nodei{i}\in \nodes}\dout\p{\nodei{i}}$ and $f < \min\limits_{\nodei{i}\in \nodes}\din\p{\nodei{i}}$ as the necessary condition (formally shown in \cref{thm:k-cast_rule}) for every node $\nodei{i}$ and any graph.

\begin{lemma}\label{thm:k-cast_rule}
  Consider a hypergraph \parsehgraph with \parseV nodes.
  To tolerate $f$ faults, a necessary condition is $f < \connectivityresult$.
\end{lemma}

\begin{proof} Let the number of Byzantine nodes be $f$, which is equal to $\min\limits_{\nodei{i}\in\nodes} \p{ \sendNd\p{\nodei{i}}, \recvNd\p{\nodei{i}} }$.
  Let \nodei{i} be the node whose \sendNd is the smallest, i.e., $\sendNd\p{\nodei{i}} = f$.
  If we corrupt the $f$ receivers of \nodei{i}, then the faulty nodes can prevent all messages sent by $\nodei{i}$ from reaching the other correct nodes.
  Similarly, if all $\recvNd\p{\nodei{i}}$ are corrupted, then Byzantine nodes can prevent other correct nodes from communicating with $\nodei{i}$.
  Therefore, in order for the hypergraph \h with \kcasts to tolerate $f$ faults, $f < \connectivityresult$ is necessary.
\end{proof}

\begin{lemma}[Connectivity necessary condition]\label{lem:nd lem}
	For any hypergraph \parsehgraph with only \kcasts and each node having at least $\Din$ incoming \kcasts and $\Dout$ outgoing \kcasts, tolerating $f$ Byzantine faults,
  $f < k\cdot \min\p{\Din,\Dout}$.
\begin{proof}
	From \cref{thm:k-cast_rule}, $f<\connectivityresult$ is necessary.
  We need to show that $\min\limits_{\nodei{i}\in\nodes}\p{\dout(\nodei{i})} \le k\cdot\Dout$ and $\min\limits_{\nodei{i}\in\nodes}\p{\din(\nodei{i})} \le k\cdot \Din$.
	Let $\Dout(\nodei{i})$ be the number of outgoing \kcasts for node $\nodei{i}$.
  By the definition of \kcasts, the independence of edges assumption (\cref{def:ind edges}), and since \h has only \kcasts, $\dout\p{\nodei{i}}\le k\cdot \Dout\p{\nodei{i}}$ must hold from the union bound.
  The equality occurs in a \h where every \kcast introduces new nodes to \nodei{i}.
  Since every node has at least \Dout outgoing \kcasts, $\min\limits_{\nodei{i}\in\nodes}\p{\Dout(\nodei{i})} = \Dout$ giving us $f<k\cdot \Dout$.
  Using similar analysis for the incoming \kcasts gives us $f<k\cdot \Din$.
  Combining the two inequalities gives $f < k\cdot \min\p{\Din,\Dout}$.
\end{proof}
\end{lemma}

The result in \cref{lem:nd lem} becomes a tight condition, already shown in literature~\cite {dolevAuthenticatedAlgorithmsByzantine1983}, if the hypergraph is treated as a regular graph using unicasts described above by setting $k = 1$, $\dout \gets d_{o}$ and $\din \gets d_{i}$.

Note that this is a necessary condition and not a sufficient condition for hypergraph connectivity.
In the rest of the paper, for ease of exposition, we assume that the graph is designed optimally so that removing any $f$ nodes and its edges from the system, the graph
still remains strongly connected, i.e.,\ is partition resistant.

\subsection{Utilizing multicasts in \protocol{}}

In order to take advantage of multicasts and the model in \protocol{}, we emulate logical full-connectivity using flooding, and appropriately choose $\delay$ parameters to ensure delivery of messages within the \delay{} bound.
The only requirement is that the hypergraph ensures $f$-connectivity, i.e., no matter which $f$ nodes are Byzantine, the remaining correct nodes can communicate with each other.

%%% Local Variables:
%%% mode: latex
%%% TeX-master: "main"
%%% End:

\section{Security Analysis}\label{sec:security}

\cref{lem: honest sender} shows that when the leader is correct, all the nodes commit the leader's blocks.

\begin{lemma}[Safety and liveness for a correct leader]\label{lem: honest sender} 
    In the steady-state, if the leader $\currentleadervar$ of view $v$ is correct and the leader proposes $\block$ for round $r$, then all correct nodes commit $\block$ for round $r$ and in view $v$.
\begin{proof}
	All correct nodes will receive $\block$ sent by $\currentleadervar$ from the bounded synchrony assumption.

	A correct node will send a \blameMsg{} in \cref{alg:protocol:blame-np} message only if any of the blame conditions are satisfied. 
	The steady state for $\currentroundvar$ begins from round $3$.
	Let a node enter any round $\currentroundvar$ at time $t$.
	The leader will enter the round at time at most $t+\delay$ whose proposed block \block{} will reach all nodes by time $t+2\delay$.
	Adding another $2\delay$ for chain synchronization, all correct nodes will obtain a valid proposal by time $t+4\delay$.
	Thus, no correct node will trigger the case in \cref{alg:protocol:blame-np} for a correct leader.

	The $f$ faulty nodes cannot generate $f+1$ view change messages to trigger view change in \cref{alg:protocol:quit-view1,alg:protocol:quit-view2}.
	Since the digital signature scheme used is secure, we can ensure a \faulty{} node cannot forge signed messages for $\currentleadervar$. 
	Therefore, all correct nodes will not blame $\currentleadervar$ through \cref{alg:protocol:blame-eq}. 

	Thus $\commitTimer(\block)$ times out in \cref{alg:protocol:commit-rule} resulting in committing of $\block$ for round $\currentroundvar$.
\end{proof}
\end{lemma}

\cref{lem: commit guarantee} shows that even if the leader is Byzantine, no two nodes can commit to different blocks in the same view $v$, guaranteeing safety within the view $v$.

\begin{lemma}[Commit safety in a view]\label{lem: commit guarantee}
	In \protocol{}, no two correct nodes can commit to different blocks $\block$ and $\block'$ at height $m$ in the same view $v$.
\end{lemma}

\begin{proof}
	We prove this by contradiction. 
	Say two correct nodes $\nodei{i}$, $\nodei{j} \in \nodes$ commit to $\block$ and $\block'$ respectively at height $m$ in view $v$.
	This corresponds to some round $r$ in view $v$.
	Without loss of generality, let $\nodei{i}$ commit $\block$ first, at time $t$.
	By time $t-4\delay$, \nodei{i} must have forwarded $\block$ while executing \cref{alg:protocol:line:voting-in-the-head}.
	This will reach $\nodei{j}$ by time $t-3\delay$ and even if it does not recognize the chain, it knows that the leader signed some other block for round $r$.
	Therefore, \nodei{j} must have committed $\block'$ before time $t+\delay$ because otherwise $\block'$ would not have been committed.
	In this case, $\nodei{j}$ must have received the block $\block'$ before time $t-3\delay$, then it would have executed \cref{alg:protocol:line:voting-in-the-head} and as a consequence $\nodei{i}$ would not have committed $\block$.
	This results in a contradiction and concludes the proof.
\end{proof}

The view change protocol is triggered when $f+1$ no progress blame messages are heard by all the correct nodes (\cref{alg:protocol:quit-view1,alg:protocol:quit-view2}). 
This indicates that at least one correct node has detected Byzantine behavior of the leader, and it is safe to quit the view $v$. 
If the leader is correct, a view change cannot occur as shown in \cref{lem: honest sender}.
In \cref{lem: proof local view change}, we show that the view change protocol for the blocking commit rule ensures safety across views.

\begin{lemma}[Unique Extensibility]\label{lem: proof local view change}
    The view change protocol ensures if a correct node commits block $\block$ in view $v$, then all correct nodes commit $\block$ in future views $v' \ge v$.
\end{lemma}

\begin{proof}
    If a block $\block$ is committed by a correct node $\nodei{i} \in \nodes$ at time $t$, then $\nodei{i}$ must have broadcast the message to all correct nodes by time $t-4\delay$ in \cref{alg:protocol:line:voting-in-the-head}. 
	\emph{All correct nodes} observed that this correctly extends their local locks \lockBlock{} (\cref{alg:protocol:line:voting-in-the-head}), and therefore did not send any \blameMsg{} messages.
	The key observation here is that the locked block of all correct nodes always extend a committed block.

	On quitting the view $v$, all correct nodes send their highest committed block to all the correct nodes (\cref{alg:line:send highest committed block}).
	Since the locked block of all correct nodes always extend a committed block, all the correct nodes vote for other node's highest committed blocks (\cref{alg:line:vote for highest qc}).
	Thus, the highest committed block among all correct nodes will always obtain a quorum certificate.
	This certificate is then multicast to all correct nodes who will update their highest committed block if it does not fork away from their local locked block \lockBlock{} (\cref{alg:line:update highest QC}).
	This is safe, since the certificate contains at least one correct node's vote, which ensures that the chain does not fork away from the highest committed block.

	The second part of the proof is to prove that a new leader can successfully start the new view mainly ensure that all the timers are correct so that \cref{lem: honest sender} can be applied from rounds $3$ onwards.

	First, we show that if the leader of the view $v'\ge v$ is correct, then it will be able to propose blocks.
	If the leaders between $(v,v')$ crash, the short-circuit of blames in \cref{alg:protocol:view-change:short-circuit} ensures another view change.

	Let the leader of view $v'$ be correct.
	In view $v$, let the first correct node observe $blameQC$ at time $t$ in \cref{alg:protocol:quit-view2}.
	By time $t+\delay$, this node will execute the procedure \Call{QuitView}, and by time $t+2\delay$, all correct nodes will quit the view $v$, and execute the procedure \Call{QuitView}.
	By time $t+6\delay$, all correct nodes will have their highest committed blocks certified.
	At this point, if a correct node had entered \Call{QuitView} at time $t+\delay$, the $5\delay$ wait was sufficient to get its highest committed block certified.

	By time $t+7\delay$, all correct nodes will receive and update their highest certified committed block and execute the procedure \Call{NewView}.
	Note that the highest committed block tracked by all correct nodes will extend the highest committed block of view $v$.

	Let the leader of the new view be the earliest node to have executed the procedure \Call{NewView} at time $t'=t+6\delay$.
	By time, $t'+\delay$, all correct nodes would have executed the procedure \Call{NewView}.
	By time $t'+4\delay$, the leader would have received the certificate for the highest committed block and also performed the chain synchronization.
	Thus, the $4\delay$ wait in \cref{alg:timer:new view leader wait for QC} suffices for a correct leader.
	The new leader proposes a valid block which will reach all the correct nodes by time $t'+5\delay$, and will be processed by all nodes (after chain synchronization) by time $t'+7\delay$.
	Even if a correct node had entered the \Call{NewView} procedure by time $t'$, the $8\delay$ time in \cref{alg:protocol:view-change:short-circuit} is sufficient to ensure that a correct leader is never blamed.

	If there are $f+1$ different commit QCs, then the new leader includes all the $f+1$ commit QCs and proposes a block extending the highest of them.
	This is safe since every QC, has a vote from at least one correct node, which ensures that the chain does not contradict with that correct node's \lockBlock{}, which ensures that it does not contradict with the highest committed block of view $v$.
	Since $f+1$ QCs are included in the status of the new proposal, at least one of these QCs must extend the highest committed block.
	The other QCs may be from Byzantine nodes trying to send lower certified blocks, but this is not a problem since the one correct node's QC will be longer than these QCs and the correct leader must extend it.
	Note that blocks lower than the highest committed block can always be certified, but the highest certified block will always extend these.

	Upon obtaining a valid proposal for $\currentroundvar = 1$, all nodes vote for the new proposal between time $t'+7\delay$ to $t'+9\delay$.
	All these votes will reach the correct leader by time $t'+10\delay$ and the leader can create the proposal for $\currentroundvar = 2$.
	By time $t'+11\delay$, all nodes receive the proposal for $\currentroundvar=2$, which will be processed by time $t'+13\delay$.
	If a correct node had sent a vote at time $t'+7\delay$, the $6\delay$ timer in \cref{alg:protocol:view change round 2 blame} ensures that a correct node is not blamed.
	From this point onwards, the system enters steady state.

	Thus, the view-change protocol ensures safety across views.
\end{proof}

Finally, we prove that \protocol{} satisfies the \textit{safety} and \textit{liveness} properties of \cref{def:SMR def} in \cref{thm:smr safety} and \cref{thm:smr liveness}.

\begin{theorem}[SMR safety]\label{thm:smr safety}
    In \protocol{}, all correct nodes commit the same block $\block$ for any height $m$.
\begin{proof}
	If a node commits $\block$ in view $v$, then \cref{lem: commit guarantee} and \cref{lem: proof local view change} ensure that the same block $\block$ is committed for height $m$ in the same or future views.
\end{proof}
\end{theorem}

\begin{theorem}[SMR liveness]\label{thm:smr liveness} 
    In \protocol{}, the protocol continues to make progress (commit blocks).
\begin{proof}
	During the steady-state, nodes commit a block as fast as the leader can propose blocks.
	During the view-change from the proof of \cref{lem: commit guarantee}, by time $21\delay$ the nodes enter the steady-state for the next view.
	Since we assume bounded synchrony within $f+1$ view-changes a new steady-state will emerge resulting in progress.
	The $f+1$ view-changes can be reduced to a constant in expectation by choosing the leaders for a view randomly~\cite{abrahamSynchronousByzantineAgreement2019}.
\end{proof}
\end{theorem}

\begin{lemma}[Communication complexity]\label{lem:communication-complexity}
	The communication complexity of \protocol{} is $O(n^2)$ during the steady state and $O(n^3)$ which is amortized $O(n^2)$ per block during the view-change.
\begin{proof}
It is easy to see that during the steady-state the communication complexity of \protocol{} is $O(n^2)$.

During the view-change, a Byzantine leader from the old view can send a long chain of size $\ell$ as its \lockBlock{}.
This results in a communication complexity of $O(\ell n^2)$.
It can also do this with $\ell$ equivocating chains still resulting in a communication complexity of $O(\ell n^2)$ since the leader will pick some version of this chain.
However, this also results in committing $\ell$ blocks in the new view, so its cost is amortized to $O(n^2)$ per committed block.
\end{proof}
\end{lemma}

%%% Local Variables:
%%% mode: latex
%%% TeX-master: "main"
%%% End:

\fi{}

\end{document}